\documentclass[conference,letterpaper]{IEEEtran}
\IEEEoverridecommandlockouts
\usepackage{amssymb}
\usepackage{graphicx}
\usepackage{amssymb,amsfonts}
\usepackage{graphicx}
\usepackage{textcomp}
\usepackage{xcolor}
\usepackage{fancyhdr}
\usepackage{soul}
\usepackage{multirow}
\usepackage{diagbox} 
\usepackage{booktabs} 
  
\usepackage[noend]{algpseudocode}  
\usepackage{algorithm}
\usepackage{amsthm}
\usepackage{url}
\usepackage{cite}
\usepackage{booktabs}
\usepackage{multirow}
\usepackage{makecell}
\usepackage{threeparttable}
\usepackage{booktabs}
\usepackage{amsmath}

\newtheorem{lemma}{Lemma}
\newtheorem{remark}{Remark}

\usepackage{subfigure}

\usepackage{hyperref}
\usepackage{bm}
\usepackage{tikz}
\usetikzlibrary{decorations.pathmorphing} 
\usetikzlibrary{fit}                    
\usetikzlibrary{backgrounds}    
\usetikzlibrary{decorations.pathreplacing,calligraphy,shadows}
\tikzstyle{microblock}=[circle,thick,minimum size=1cm,draw=gray!80,fill=gray!20]
\tikzstyle{microblockH}=[circle,thick,minimum size=1cm, dashed, minimum size=0.8cm,draw=gray!80,fill=gray!20]

\tikzstyle{keyblockH}=[rectangle,thick,minimum size=1cm,draw=blue!80,fill=blue!20]
\tikzstyle{keyblockA}=[rectangle,thick,minimum size=1cm,draw=red!80,fill=red!20]
\tikzstyle{keyblockH_potential}=[rectangle,thick,dashed,minimum size=1cm,draw=blue!20,fill=blue!10]
\tikzstyle{keyblockA_potential}=[rectangle,thick,dashed,minimum size=1cm,draw=red!20,fill=red!10]

\usepackage{enumitem}
\setenumerate[1]{itemsep=0pt,partopsep=0pt,parsep=\parskip,topsep=5pt}
\setitemize[1]{itemsep=0pt,partopsep=0pt,parsep=\parskip,topsep=5pt}
\setdescription{itemsep=0pt,partopsep=0pt,parsep=\parskip,topsep=5pt}

\def\BibTeX{{\rm B\kern-.05em{\sc i\kern-.025em b}\kern-.08em
    T\kern-.1667em\lower.7ex\hbox{E}\kern-.125emX}}

\begin{document}

\title{Incentive Analysis of Bitcoin-NG, Revisited\\}

\author{
\IEEEauthorblockN{Jianyu Niu\IEEEauthorrefmark{1}, Ziyu Wang\IEEEauthorrefmark{2}, Fangyu Gai\IEEEauthorrefmark{1}, and Chen Feng\IEEEauthorrefmark{1}}
\IEEEauthorblockA{\IEEEauthorrefmark{1}School of Engineering, University of British Columbia (Okanagan Campus)}
\IEEEauthorblockA{\IEEEauthorrefmark{2}School of Cyber Science and Technology, Beihang University}
\IEEEauthorblockA{\IEEEauthorrefmark{1}{\{jianyu.niu, fangyu.gai, chen.feng\}@ubc.ca, \IEEEauthorrefmark{2}{wangziyu@buaa.edu.cn}}}
}

\maketitle
\cfoot{}
\pagestyle{fancy}
\rfoot{\thepage}

\begin{abstract}
Bitcoin-NG is among the first blockchain protocols to approach the \emph{near-optimal} throughput by decoupling blockchain operation into two planes: leader election and transaction serialization. Its decoupling idea has inspired a new generation of high-performance blockchain protocols. However, the existing incentive analysis of Bitcoin-NG has several limitations. First, the impact of network capacity is ignored. Second, an integrated incentive analysis that jointly considers both key blocks and microblocks is still missing. 

In this paper, we aim to address these two limitations. First, we propose a new incentive analysis that takes the network capacity into account, showing that Bitcoin-NG can still maintain incentive compatibility against the microblock mining attack even under limited network capacity.
Second, we leverage a Markov decision process (MDP) to jointly analyze the incentive of both key blocks and microblocks, showing that the selfish mining revenue of Bitcoin-NG is a little higher than that in Bitcoin only when the selfish miner controls more than $35\%$ of the mining power.
We hope that our in-depth incentive analysis for Bitcoin-NG can shed some light on the mechanism design and incentive analysis of next-generation blockchain protocols.
\end{abstract}

\section{Introduction} \label{sec:introduction}
Bitcoin---the current largest and most influential cryptocurrency---has sparked many other cryptocurrencies like Ethereum~\cite{buterin2014next} and Litecoin~\cite{reed2017litecoin}, gaining much attention from both academia and industry \cite{nakamoto2012bitcoin}. 
The key innovation behind Bitcoin is \emph{Nakamoto Consensus} (NC), which is used to realize a distributed ledger known as a blockchain. Blockchains have
unique features of decentralization, security, and privacy, making them a fundamental trust infrastructure for supporting various future decentralized Internet applications, ranging from IoT, health care, supply chain management, to clean energy~\cite{mettler2016blockchain, iot2019}.

Despite the popularity, Bitcoin and other NC-based blockchains have suffered from low throughput (e.g., $7$ TPS\footnote{TPS is short for transactions per second.} in Bitcoin) and poor network utilization (e.g., less than $2\%$ in Bitcoin~\cite{Decker2013}).
The low throughput of Bitcoin is mostly due to its choice of two system parameters: small block size (originally $1$ MB) and long block interval (on average $10$ minutes). Although increasing the block size or shortening the block interval can increase the throughput,
this reduces the security level of Bitcoin because forks are more likely to occur \cite{garay2015bitcoin, Pass2017, ghost}. 
Indeed, it has been shown in various studies~\cite{CDE+-16, bitcoinng, ghost} that redesigning the underlying NC (rather than fine-tuning the system parameters) is essential to improve the throughput without sacrificing security.

Bitcoin-NG (Next Generation) ~\cite{bitcoinng} is among the first and the most prominent NC-based blockchains to approach the \emph{near-optimal} throughput.
Bitcoin-NG creatively employs two types of blocks: $1$) a \emph{key block} that is very similar to a conventional block in Bitcoin except that it doesn't carry any transactions, and $2$) a \emph{microblock} that carries transactions. Every key block is generated through the leader election process (often known as the mining process) in NC, and the corresponding leader will receive a block reward (if its key block ends up in the longest chain). In addition, this leader can issue multiple microblocks and receive the transaction fees until the next key block is generated. Unlike Bitcoin, Bitcoin-NG decouples leader election and transaction serialization. Intuitively, it is this decoupling that enables Bitcoin-NG to approach the near-optimal throughput, since the microblocks can be produced at a rate up to the network capacity. Perhaps for this reason, Bitcoin-NG has been adopted by two cryptocurrencies: Waves\footnote{Waves: https://docs.wavesplatform.com/} and Aeternity\footnote{Aeternity: https://aeternity.com/}.

More importantly, this decoupling idea has inspired a new generation of blockchain protocols including ByzCoin~\cite{ByzCoin}, Hybrid consensus~\cite{pass20171}, Prism~\cite{bagaria2019deconstructing}, and many others. Although these protocols are able to achieve lower latency and/or higher throughput than Bitcoin-NG, their incentive mechanism design and analysis still remain unclear. 
{Such incentive analysis is particularly important for understanding incentive-based attacks, in which all the nodes are assumed to be rational and profit driven.}
Nevertheless, even the existing incentive analysis of Bitcoin-NG has several limitations, as we will explain shortly. As a starting point to bridge this research gap, we aim to provide an in-depth incentive analysis for Bitcoin-NG, hoping that it would shed some light on the mechanism design and incentive analysis of aforementioned next-generation blockchain protocols.

The prior work of Bitcoin-NG found that Bitcoin-NG cannot maintain the incentive compatibility\footnote{The expected relative revenue of a miner should be proportional to its mining power.} of microblocks when an adversary controls more than $29\%$ of the total computation power \cite{bitcoinng, yin18bngrrevisit}. In addition, an adversary in Bitcoin-NG can gain a higher share of block reward than in Bitcoin, making Bitcoin-NG more vulnerable~\cite{ziyu2019}. Despite these important findings, previous incentive analysis of Bitcoin-NG has the following limitations. First, previous analysis completely ignores the impact of network capacity \cite{bitcoinng, yin18bngrrevisit, ziyu2019}. \emph{How can we take into account the network capacity constraints}?
Second, previous analysis mostly focuses on microblocks. (See, e.g., Sec.~\ref{sec:original analysis} for details.) \emph{How can we take into account the effect of key blocks}?

To answer the first question, we develop a new probabilistic analysis that takes network capacity into account. In particular, we model the interval between two consecutive key blocks as an exponential random variable and introduce the generation rate of microblocks to capture the impact of network capacity. Then, we apply the Chernoff-type bounding techniques to derive the long-term average revenue of the adversary. We find that by choosing suitable system parameters, Bitcoin-NG can still maintain incentive compatibility even under network capacity constraints. In other words, introducing network capacity constraints doesn't make it harder to maintain the incentive compatibility. More specifically, when the adversary controls less than $29\%$ of the mining power, the incentive compatibility of Bitcoin-NG can be maintained for all types of transactions. When the adversary controls more than $29\%$ of the mining power, the incentive compatibility can be maintained for regular transactions but not for whale transactions with high fees.

To address the second question, we leverage a Markov decision process (MDP) model to jointly analyze the incentive of key blocks and microblocks. Although similar analysis has been conducted by Sapirshtein et al.~\cite{sapirshtein2016optimal} in the context of Bitcoin\footnote{Due to the similarity, the MDP can be directly used to model the key-block mining in Bitcoin-NG.}, the microblock structure in Bitcoin-NG introduces additional complexity for the MDP design (e.g., more mining strategies and rewards). 
To make the MDP tractable, we confine our analysis to a family of selfish mining strategies\footnote{This family is broader than that in the existing analysis, as discussed in Sec.~\ref{sec:comp2}.}. Our results show that the optimal selfish mining revenue in Bitcoin-NG is just a little higher than that in Bitcoin when the selfish computation power is greater than $35\%$.

\noindent \textbf{Contributions:} The contributions of this paper are summarized as follows:
\begin{itemize}[leftmargin=*]
\item We propose a new incentive analysis of Bitcoin-NG considering the network capacity constraints. 
Our results show that Bitcoin-NG can still maintain incentive compatibility  against the microblock mining attack.

\item We model the selfish mining of key blocks and microblocks jointly into an MDP. 
Our results show that the selfish mining revenue in Bitcoin-NG is a little higher than that in Bitcoin only when the selfish mining power $\alpha$ is greater than $35\%$. 

\item We show the distribution of transaction fees by scanning transactions in a recent block history of Bitcoin, which supports our assumptions in our system model. 
\end{itemize}

\section{Background} \label{sec:background}
\subsection{A Primer on Bitcoin}
Bitcoin relies on Nakamoto Consensus (NC) to make a group of distributed and mutually distrusting participants reach agreement on a transparent and immutable ledger, also known as the blockchain. A blockchain is a list of blocks linked by hash values with each block containing a batch of ordered transactions. To make all participants agree on the same chain of blocks, NC leverages two components: the Proof-of-Work (PoW) mechanism and the longest chain rule (LCR). Each participant in NC (also referred to as a miner) collects valid and unconfirmed transactions from the network, orders and packs these transactions into a block. In addition, a valid block needs to contain a proof of work, i.e., its owner needs to find a value of the nonce (i.e., a changeable data filed) such that the hash value of this block has required leading zeros~\cite{nakamoto2012bitcoin}. The length of leading zeros is also known as the mining difficulty, which can be tuned by the system so that new blocks are mined every ten minutes on average. 

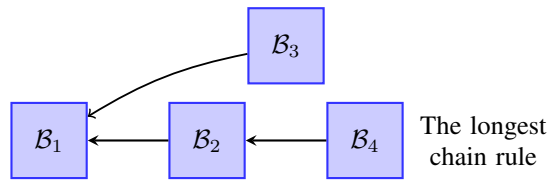
\begin{figure}[!ht]
\centering
\begin{tikzpicture}[x=0.7cm,y=0.7cm]
    \node (k1) at (0 , 1) [keyblockH] {$\mathcal{B}_1$};
    \node (k2) at (3 , 1) [keyblockH] {$\mathcal{B}_2$};
    \node (k3) at (4.5 , 2.8) [keyblockH] {$\mathcal{B}_3$};
    \node (k4) at (6 , 1) [keyblockH] {$\mathcal{B}_4$};
    \node (k5) at (8.25 , 1) [text width=2cm,align=center] {The longest\\chain rule};
    \draw [thick,->,>=stealth] (k2) -- (k1);
    \draw [line width=0.7pt,->, bend left=-10] (k3) edge (k1);
    \draw [thick,->,>=stealth] (k4) -- (k2);
\end{tikzpicture}
\caption{\textbf{An illustration of the chain structure in Bitcoin.}}
\label{fig:bitcoin}
\end{figure}

Once a new block is produced, it will be immediately broadcast to the entire network. 
Ideally, the block should be accepted by all participants before the next block is produced.
In reality, two new blocks might be mined around the same time, leading to a fork in which two ``child'' blocks share a common ``parent" block. 
To resolve such a fork, an honest miner always accepts the longest chain as the valid one. 
See Fig.~\ref{fig:bitcoin} for an illustration. 
Block $B_3$ is a forking block, which will be abandoned by the honest miners according to the longest chain rule. 
In Bitcoin, a block miner will receive a block reward (if its block is eventually included in the longest chain) as well as transaction fees as another type of reward. These incentives encourage miners to devote their computational resources to the system. 

\noindent \textbf{Selfish Mining.} Although NC is designed to fairly reward miners according to their contributions to the system (i.e., miners' revenue is proportional to their devoted computation power), the studies in~\cite{eyal2014majority, sapirshtein2016optimal, nayak2016stubborn, gervais2016security} show that a selfish miner can gain more revenue than its fair share by deviating from the protocol. This mining attack is called \emph{selfish mining}. In this attack, a selfish miner can keep its newly generated blocks secret, mine on top of these blocks, and create forks on purpose when necessary. In particular, when some honest miner generates a new block, a selfish miner will publish one secret block to match this honest block as a competition or publish two blocks to override this honest block because honest miners follow LCR. 
A particular example is depicted in Fig.~\ref{fig:bitcoin}. 
A selfish miner successfully mines two consecutive blocks $B_2$ and $B_4$ in advance, and it then receives a block $B_3$ mined by some honest miners. It will immediately publish its two blocks. In this way,
it will not only get two block rewards for $B_2$ and $B_4$ but also make the honest block $B_3$ abandoned by all the miners.
In other words, a selfish miner can obtain higher revenue than it deserves by invalidating honest blocks. Indeed,
Sapirshtein et al. have used an MDP to model various selfish mining strategies and concluded that the optimal computation power threshold making selfish mining profitable is $23.21\%$ \cite{sapirshtein2016optimal}.

\subsection{A Primer on Bitcoin-NG} \label{subsec:bitcoinng}
In Bitcoin, the mining of blocks has two functionalities: $1$) electing leaders (i.e., the owners of valid blocks) by NC, and $2$) ordering and verifying transactions. 
By differentiating block functionalities, Bitcoin-NG decouples the leader election with the transaction serialization. Specifically, Bitcoin-NG uses key blocks mined through PoW to elect a leader at a stable rate (e.g., one key block per 100 seconds). 
Each leader can produce several microblocks containing unconfirmed transactions at another rate, often higher than the key block rate (e.g., one microblock per 20 seconds).
In a nutshell, a key block is very similar to a conventional block in Bitcoin except that it does not carry any transactions. On the other hand, microblocks contain transactions but do not contain any proof of work. 
Although the rate of microblocks is usually much larger than the key block generation rate, it has to be bounded
in order to prevent adversarial leaders from swamping the system with microblocks.
This decoupling enables Bitcoin-NG to process many microblocks between two consecutive key blocks, and significantly increasing its transaction throughput. Fig.~\ref{fig:bitcoinng} illustrates these two types of blocks. 
\begin{figure}[!ht]
\centering
\begin{tikzpicture}[x=0.8cm,y=0.8cm]
    \node (k1) at (0 , 1) [keyblockH] {$\mathcal{B}_j$};
    \node (m1) at (1.8, 1) [microblock] {$\sigma_{j}^{1}$};
    \node (m2) at (3.6, 1) [microblock] {$\sigma_{j}^{2}$};
    \node (k2) at (5.4, 1) [keyblockH] {$\mathcal{B}_{j+1}$};
    \node (m7) at (7.2, 1) [microblock] {$\sigma_{\scriptscriptstyle j+1}^{\scriptscriptstyle 1}$};
    \node (m8) at (9, 1) [microblock] {$\sigma_{\scriptscriptstyle j+1}^{\scriptscriptstyle 2}$};
    \node (k3) at (3.6, 2.5) [keyblockH] {$\mathcal{B}_{j+2}$};
    \node (m5) at (5.4, 2.5) [microblock] {$\sigma_{j+2}^{1}$};
    \draw [thick,->,>=stealth] (m1) -- (k1);
    \draw [thick,->,>=stealth] (m2) -- (m1);
    \draw [thick,->,>=stealth] (k2) -- (m2);
    \draw [thick,->,>=stealth] (m7) -- (k2);
    \draw [thick,->,>=stealth] (m8) -- (m7);
    \draw [thick,->,>=stealth] (k3) -- (m1);
    \draw [thick,->,>=stealth] (m5) -- (k3);
\end{tikzpicture}
\caption{\textbf{An illustration of Bitcoin-NG.} A square (respectively, circle) block denotes the key block (respectively, microblock). The microblocks are issued by the three key-block miners $\mathcal{B}_{j},\mathcal{B}_{j+1},\mathcal{B}_{j+2}$ with their signatures $\sigma_{j},\sigma_{j+1},\sigma_{j+2}$, respectively.}
\label{fig:bitcoinng}
\end{figure}
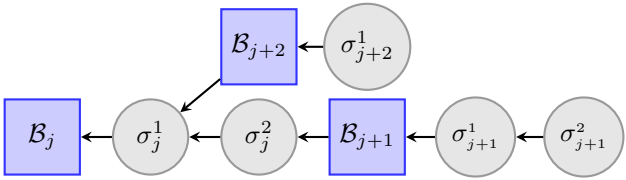

Bitcoin-NG adopts a similar fork choice rule as Bitcoin.  
In Bitcoin-NG, microblocks carry no weight, not even a secondary index for miners to choose which key block to mine. 
For instance, in Fig.~\ref{fig:bitcoinng}, there are two forking branches with the same number of key blocks but different numbers of microblocks. 
However, honest miners treat these two forking branches as equal and adopt a uniform tie-breaking rule to choose one branch~\cite{bitcoinng}.
Honest miners then mine on the latest microblock in this branch.
In a nutshell, an honest miner still follows LCR to choose a ``right'' key block (i.e., the last key block in the longest chain only consisted of key blocks), and then mine on the latest microblock produced by the key-block miner. 
Thus, without microblocks, the mining process of key blocks is the same as the one in Bitcoin. 
In fact, the selfish mining attack in Bitcoin can be used here to attack key blocks in Bitcoin-NG.
However, as Bitcoin-NG introduces additional microblocks, there are some new possible attacks for microblocks (which will be introduced shortly).

\subsection{Transaction Fee Distribution} \label{sec:discussion}
{Transaction fee is used to incentivize miners to include transactions in their blocks. Therefore, the higher the transaction fee is, the more miners try to include the transaction into the latest block. Fig.~\ref{fig:btcfeedist} shows the practical transaction fee distributions by scanning blocks with height from $627195$ to $627894$ in Bitcoin. The results show that about $77.8\%$ transactions have a quite small fee (less than $0.0001$ BTC). When relaxing to $0.0005$ BTC, the proportion accounts for around $98.5\%$. 
This findings supports our later transaction fee model: the vast majority of transactions in a blockchain system have small transaction fee.}

\begin{figure}[!ht]
\centering
\includegraphics[width=0.5\linewidth]{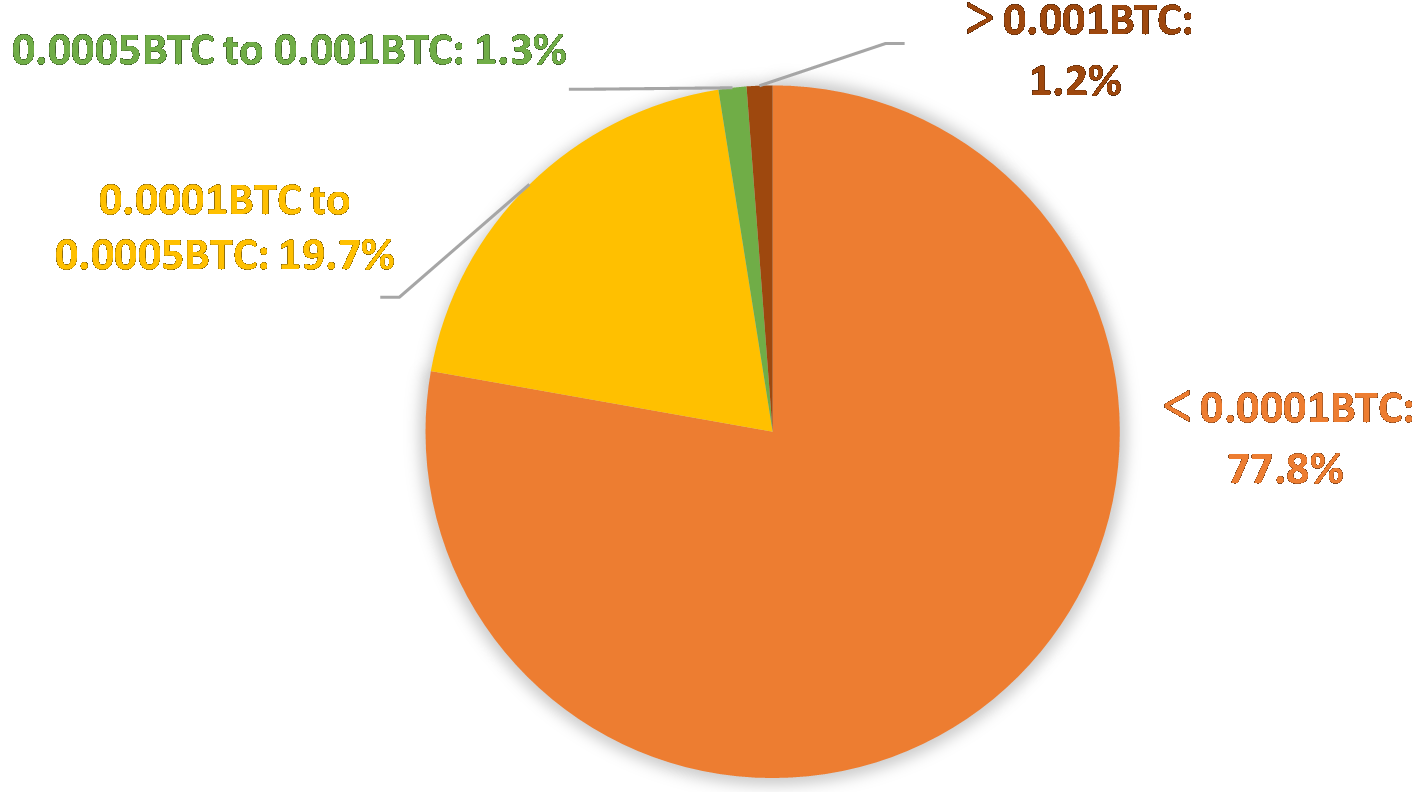}
\caption{\textbf{The transaction fee distribution in the blocks ranging from block height $627195$ to $627894$.}}
\vspace{-2mm}
\label{fig:btcfeedist}
\end{figure}

\section{System Model} \label{sec:model}
\subsection{Mining Model} \label{sec:mining model}
Following the mining models~\cite{eyal2014majority, bitcoinng, niu2019selfish}, we consider a collection of $n$ miners, denoted by the set $\mathcal{N}$.  
We assume a subset of miners $\mathcal{S} \subset \mathcal{N}$ are selfish and can deviate from the protocol to maximize their profit. The other miners in $\mathcal{N} \setminus \mathcal{S}$ are honest and follow the protocol. The $i$-th miner has $m_{i}$ fraction of the total hash power. In addition, the selfish (respectively, honest) miners control $\alpha$ (respectively, $\beta$) fraction of total mining power. That is, $\alpha = \sum_{i \in \mathcal{S}} m_i$ and $\beta = \sum_{i \in \mathcal{N} \setminus \mathcal{S}} m_i$. Clearly, $\alpha + \beta = 1$. 
We assume that all the selfish miners form a mining pool controlled by a single player, which is referred to as the selfish miner. 

The mining process of the key block can be modeled as a Poisson process with rate $f$, as shown in \cite{nakamoto2012bitcoin, bagaria2019deconstructing}\footnote{The key block is mined by solving the PoW puzzle, and the value of $f$ can be calculated by the average block interval, which is $10$ minutes in Bitcoin and $100$ seconds per block in Bitcoin-NG.}.
The key block mining process of the $i$th miner is also a Poisson process with the rate $m_i f$. Hence, the selfish miner generates key blocks at rate $\alpha f$, and the remaining honest miners (as a whole) generate key blocks at rate $\beta f$. In addition, the miner of each key block becomes a leader and can issue a series of microblocks containing as many transactions as possible (up to the maximum microblock size) at a constant rate $v$ until the next key block is mined (see Sec~\ref{subsec:bitcoinng}). 
Specifically, a block (including key block and microblock) mined by an honest (respectively, the selfish) miner is referred to as \emph{honest} (respectively, \emph{selfish}) block.

\subsection{Network Model}
Following the network model of Bitcoin~\cite{eyal2014majority, niu2019selfish}, we assume that honest miners are fully connected through the underlying network, and an honest miner spends negligible time to broadcast a key block or microblock in Bitcoin-NG\footnote{This assumption is reasonable for key blocks because the inter-arrival time of two consecutive key blocks is often much larger than the block propagation delay. On the other hand, this assumption can be relaxed for microblocks, as we will show later.}. In addition, we assume that the selfish miner can broadcast its private blocks immediately after it sees a new honest key block.

\subsection{Mining Rewards} \label{sec:mining rewards}
In Bitcoin-NG, there are two types of rewards, namely key-block reward and transaction fee. Every miner obtains a key-block reward if it mines a key block by successfully solving a PoW puzzle and its key block ends up in the longest chain. 
In addition, each transaction has a fee (i.e., transaction fee) as a reward to incentive a key-block miner to verify and execute this transaction.
This fee encourages miners to include as many transactions as possible (up to the microblock size limit). 
For simplicity, we assume two types of transactions according to their transaction fees: ``whale" transactions with a high fee and regular transactions with a low fee. Also, we assume that the vast majority of transactions are regular ones. These assumptions are made based on the fee distribution discussed in Sec.~\ref{sec:discussion}. 
Next, we assume that the transaction size is fixed, and so the maximum number of transactions included in a microblock is also fixed.  
We also assume that miners have enough pending transactions to be included in microblocks\footnote{This assumption is reasonable in Bitcoin and Ethereum-like public blockchains. For instance, a mempool visualization website~\cite{transactionpool} shows that the number of pending transactions is currently around $34,000$, which is about $100$ blocks (e.g., $1$ Mb block size) worth of transfer.}. 

We call a microblock regular if it contains only regular transactions. In addition, we refer to the total transaction fees included in a regular microblock as the microblock fee and use $R_t$ to denote it. 
In addition, we use $R_b$ to denote the key-block reward. Let $k = R_b/R_t$ denote the ratio of the block reward to the microblock fee. This ratio $k$ ranges from $(0, \infty)$. When $k$ approaches $0$ (respectively, $\infty$), it implies that the transaction fee (respectively, key-block reward) dominates the reward. The different values of $k$ exhibit the various impact of rewards on the Bitcoin-NG system.

\begin{figure}[!ht]
\centering
\begin{tikzpicture}[x=0.8cm,y=0.8cm]
    \node (k1) at (0 , 1) [keyblockA] {$\mathcal{B}_j$};
    \node (m1) at (1.75 , 2) [microblock] {$\sigma_j$};
    \node (m3) at (3.5 , 2) [fill=none] {$\cdots$};
    \node (m4) at (5 , 2) [microblock] {$\sigma_j$};
    \node (k2) at (6.5 , 2) [keyblockA_potential] {$\mathcal{B}_{j+1}$};
    \node (k3) at (6.5 , 0.5) [keyblockH_potential] {$\mathcal{B}_{j+1}$};
    \draw [line width=0.7pt,->,dashed,bend left=10] (k3) edge (m1);
    \draw [thick,->,>=stealth] (m1) -- (k1);
    \draw [thick,->,>=stealth] (m3) -- (m1);
    \draw [thick,->,>=stealth] (m4) -- (m3);
    \draw [dashed,->,>=stealth] (k2) -- (m4); 
    \node at (7, 2) [draw=none, anchor=west] {selfish key block};
    \node at (7, 0.5) [draw=none, anchor=west] {honest key block};
    \node at (0.75 , 0.5) [draw=none, anchor=west] {selfish key block};
\end{tikzpicture}
\caption{\textbf{An example of the transaction inclusion attack.} The first two microblocks after the selfish $\mathcal{B}_j$ have been published and so they are public to honest miners. The other microblocks are kept private. A dashed square block denotes a future mined block.} 
\label{fig:txinclusionatt}
\end{figure}
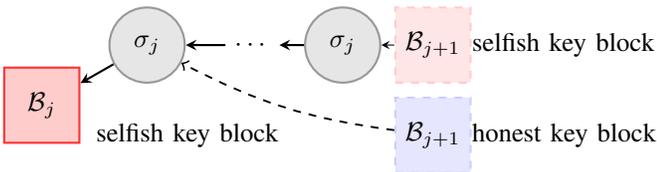

\subsection{Mining Strategies} \label{sec:mining strategy}
In this section, we introduce the mining strategies for honest and selfish miners. 
In particular, we focus on mining strategies for microblocks, since strategies for key blocks have been extensively studied~\cite{eyal2014majority,sapirshtein2016optimal,nayak2016stubborn}. 
In Bitcoin-NG, an honest key-block miner includes transactions in microblocks and publishes these microblocks to win transaction fees. This is called the transaction inclusion rule. In addition, an honest miner should accept as many microblocks issued by the previous key-block miner as possible and mine on the latest received microblock, i.e., obeying the longest chain extension rule. 
By contrast, a selfish miner could break the transaction inclusion and the longest chain extension rules to maximize its profit as explained below:
\begin{itemize}[leftmargin=*]
\item \textbf{Transaction inclusion attack.} 
When the selfish miner publishes one key block and generates multiple microblocks, it keeps the last several microblocks private. 
That is, the selfish miner continues to mine on top of its latest microblock chain, while honest miners can only mine on top of the last published microblock. 
Fig.~\ref{fig:txinclusionatt} shows the case in which the selfish miner withholds some of its microblocks mined after the key block $\mathcal{B}_j$, and honest miners mine on the last public microblock of the selfish miner. 
This attack is incentivized if transaction fees in microblocks go primarily to the next key-block owner. (See the transaction fee distribution in Sec~\ref{sec:original analysis}.)

\item \textbf{Longest chain extension attack.} 
When the selfish miner adopts an honest key block, it can reject some (or all) microblocks and mine directly on the last accepted microblock block (or the last key block, respectively). In other words, the selfish miner rejects the transactions in these microblocks issued by the previous honest key-block miner. 
This attack is illustrated in Fig.~\ref{fig:lchainextatt}. 
This attack is incentivized if transaction fees go primarily to the current key-block owner.

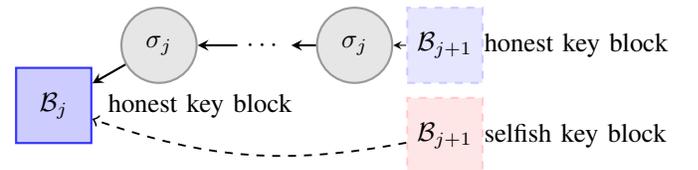
\begin{figure}[!ht]
\centering
\begin{tikzpicture}[x=0.8cm,y=0.8cm]
    \node (k1) at (0 , 1) [keyblockH] {$\mathcal{B}_j$};
    \node at (0.75,1) [draw=none, anchor=west] {honest key block};
    \node (m2) at (1.75, 2) [microblock] {$\sigma_j$};
    \node (m3) at (3.5 , 2) [fill=none] {$\cdots$};
    \node (m4) at (5, 2) [microblock] {$\sigma_j$};
    \node (k2) at (6.5  , 2) [keyblockH_potential] {$\mathcal{B}_{j+1}$};
    \node at (7 , 2) [draw=none, anchor=west] {honest key block};
    \node (k3) at (6.5  , 0.5) [keyblockA_potential] {$\mathcal{B}_{j+1}$};
    \node at (7 , 0.5) [draw=none, anchor=west] {selfish key block};
    \draw [line width=0.7pt,->,dashed,bend left=15] (k3) edge (k1);
    \draw [thick,->,>=stealth] (m2) -- (k1);
    \draw [thick,->,>=stealth] (m3) -- (m2);
    \draw [thick,->,>=stealth] (m4) -- (m3);
    \draw [dashed,->,>=stealth] (k2) -- (m4); 
\end{tikzpicture}
\caption{\textbf{An example of the longest chain extension attack.} The selfish miner rejects all the microblocks and mines it key block on top of the honest $\mathcal{B}_j$. A dashed square block denotes a future mined block.}
\label{fig:lchainextatt}
\end{figure}
\end{itemize}

\subsection{Mining Revenue} \label{sec:mining revenue}
The selfish miner is incentivized to find an optimal selfish mining strategy to increase its revenue. Specifically, the utility of the selfish miner can be defined as its relative revenue~\cite{eyal2014majority,sapirshtein2016optimal}, i.e., the ratio of the selfish miner's revenue (including key-block reward and transaction fees) to all miners' revenue. In other words, the selfish miner would like to increase its share of key-block reward and transaction fees generated by the system, which is given by 
\begin{equation}
\label{eq:revenue}
    u = \lim_{t \rightarrow \infty} \frac{r_{a}(t) + t_{a}(t)}{r_{a}(t) + r_{h}(t) + t_{a}(t) + t_{h}(t)}.
\end{equation}
Here, $r_{a}(t)$ and $t_{a}(t)$ are the key-block rewards and transaction fees won by selfish miners during time $[0,t]$, respectively. . Similarly, $r_{h}(t)$ and $t_{h}(t)$ are the key-block rewards and transaction fees won by honest miners. Note that all of them are random variables. As $t \to \infty$, the ratio in \eqref{eq:revenue} converges by the law of large numbers.

\section{Incentive Analysis for Microblock}  \label{sec:original analysis}
In this section, we present the prior incentive analysis of microblocks \cite{bitcoinng, yin18bngrrevisit}. 
In particular, the analysis does not consider the selfish mining of key blocks. That is, it assumes that the selfish miner always adopts honest miners' key blocks and immediately publishes its new key blocks. 
This assumption is justified shortly and will be relaxed by considering the joint mining of microblocks and key blocks in Sec.~\ref{sec:MDP}. 
In addition, as the propagation delay of key blocks is negligible, forked key blocks are also not considered. 

To resist the transaction inclusion attack and the longest chain extension attack, Bitcoin-NG divides the transaction fees included in microblocks between two consecutive key-block miners into two parts. The first key-block miner gets the $r$ fraction ($r \in [0, 1]$), while the second one obtains the remaining $1-r$ fraction. Fig.~\ref{fig:btcng} illustrates this fee distribution rule. The remaining subsections explain how to decide the value of $r$ to resist the microblock mining attacks. 
\begin{figure}[!ht]
\centering
\begin{tikzpicture}[x=1cm,y=1cm]
    \node (k1) at (0 , 1) [keyblockH] {$\mathcal{B}_j$};
    \node (m1) at (1.5 , 1) [microblock] {$\sigma_j$};
    \node (m2) at (3 , 1) [microblock] {$\sigma_j$};
    \node (m3) at (4.25 , 1) [fill=none] {$\cdots$};
    \node (m4) at (5.5 , 1) [microblock] {$\sigma_j$};
    \node (k2) at (7 , 1) [keyblockH] {$\mathcal{B}_{j+1}$};
    \draw [line width=0.7pt,->,dashed,bend right=20] (3, 2.1) edge node[anchor=south] {$r$} (k1.north);
    \draw [line width=0.7pt,->,dashed,bend left=20] (4, 2.1) edge node[anchor=south] {$1-r$} (k2.north);
    \draw [thick,->,>=stealth] (m1) -- (k1);
    \draw [thick,->,>=stealth] (m2) -- (m1);
    \draw [thick,->,>=stealth] (m3) -- (m2);
    \draw [thick,->,>=stealth] (m4) -- (m3);
    \draw [thick,->,>=stealth] (k2) -- (m4);
    \draw[decorate,decoration={calligraphic brace,amplitude=2mm},ultra thick] (5.5,0.4)--(0,0.4);
    \node (note1) at (2.75, 0) [rectangle, draw=none,  minimum width=1.1cm, minimum height = 0.5cm] {The miner of the key block $\mathcal{B}_j$ produces};
    \node (note2) at (2.75, -0.5) [rectangle, draw=none,  minimum width=1.1cm, minimum height = 0.5cm] {microblocks utilizing its signature $\sigma_j$.};
    \draw[decorate,decoration={calligraphic brace,amplitude=2mm},ultra thick] (1.5,1.6)--(5.5,1.6);
    \node (note3) at (3.5, 2.1) [rectangle, draw=white, fill=white, minimum width=1.1cm, minimum height = 0.5cm] {transaction fees};
\end{tikzpicture}
\caption{\textbf{Bitcoin-NG fee distribution rule}.}
\label{fig:btcng}
\end{figure}
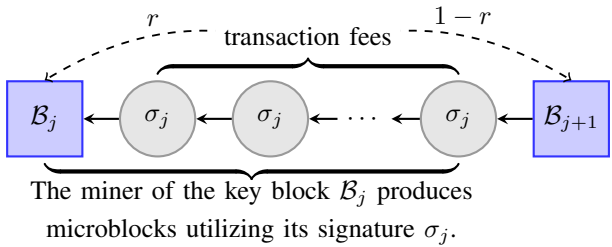

\subsection{Resisting Transaction Inclusion Attack}  \label{subsec:attack1}
Recall from Sec.~\ref{sec:mining strategy} that the selfish miner can withhold a microblock to avoid sharing its transaction fees with the subsequent key-block miner. (We refer readers to the original paper \cite{bitcoinng} for more details.). Note that the probability for the selfish (respectively, honest) miner mines a block is $\alpha$ (respectively, $\beta$ = $1 - \alpha$).
To guarantee the average revenue of the selfish miner launching the above attack is smaller than what it deserves, the distribution ratio $r$ should satisfy
\begin{equation}
    \overset{\text{win 100\%}}{\overbrace{\alpha \times 100\%}} + \overset{\text{Lose 100\%, but mine after txn}}{\overbrace{(1-\alpha ) \times \alpha \times (100\% - r)}} < r,
\end{equation}
therefore $r>1-\frac{1-\alpha}{1+\alpha-\alpha^2}$.
This ratio requirement encourages the selfish miner to place a transaction in a public microblock. 

Later, Yin et al.\cite{yin18bngrrevisit} found that the above computation neglects a case: the incumbent leader can be re-elected as the next leader and gain an extra $\alpha (1 - r)$ fraction of the transaction fee. Thus, the distribution ratio $r$ should satisfy
\begin{equation}
    \overset{\text{win 100\%}}{\overbrace{\alpha \times 100\%}} + \overset{\text{Lose 100\%, but mine after txn}}{\overbrace{(1-\alpha ) \times \alpha \times (100\% - r)}} < r + \alpha (1 - r),
\end{equation}
therefore $r>\frac{\alpha}{1-\alpha}$.
\subsection{Resisting Longest Chain Extension Attack} \label{subsec:attack2}
To increase revenue from some transactions, the selfish miner can ignore these transactions in an honest microblock and mine on a previous microblock. Later on, if the selfish miner mines a key block, it can place these transactions in its own microblock.
To resist this attack, the selfish miner's revenue in this case must be smaller than the revenue obtained by obeying the longest chain extension rule. 
Therefore, we have
\begin{equation} \label{eq:longest}
    \overset{\text{Mine next key block}}{\overbrace{\alpha \times r}} + \overset{\text{Mine the third key Block}}{\overbrace{ \alpha^2 \times (100\% - r)}} < \overset{\text{Mine on microblock}}{\overbrace{\alpha (100\% - r)}},
\end{equation}
which leads to $r < \frac{1-\alpha}{2-\alpha}$. 
Taking the upper bound into consideration, the distribution ratio $r$ satisfies $1-\frac{1-\alpha}{1+\alpha-\alpha^2} < r < \frac{1-\alpha}{2-\alpha}$.
In particular, when $\alpha$ is less than $25\%$, we obtain $37\% <r < 43\%$. Hence, $r = 40\%$ is chosen in the Bitcoin-NG~\cite{bitcoinng}. 

\subsection{Limitations of Existing Analysis} \label{subsec:limitation}
We now describe some limitations of the previous analysis, which can help us to better understand the differences between the above analysis and our analysis in the following section.

Before giving detailed descriptions, let us replay the longest chain extension attack, as shown in Fig.~\ref{fig:limit}. We make two simplifications to better illustrate the analysis limitation: $1$) each leader is allowed to only create one microblock; $2$) each microblock is allowed to only contain one transaction. 
Consider a scenario where an honest miner produces a key block $B_j$ as well as a microblock containing a transaction $tx$. 
If the selfish miner obeys the longest chain extension rule and finds the next key block with probability $\alpha$, it will get a $1-r$ fraction of the transaction fee (which corresponds to the last item in Equation~\eqref{eq:longest}).
However, the selfish miner can directly mine on the key block $B_j$, hoping to win a higher transaction fee of $tx$.
If the selfish miner happens to create the next key block $B_{j+1}$ with probability $\alpha$, it can win $r$ fraction of the transaction fee by including $tx$ in its own microblock (which corresponds to the first item in Equation~\eqref{eq:longest}). If the selfish miner is lucky to mine the next consecutive key block $B_{j+2}$, it will win the remaining $1-r$ of the transaction fee.
Combining all conditions together leads to Equation~\eqref{eq:longest}.

The above simple analysis is quite reasonable if the selfish miner just hopes to get a higher fee from a targeted transaction $tx$. 
However, in reality, the selfish miner usually aims to increase its revenue from all transactions instead of just a targeted one.
From the above discussion, we can see that if the selfish miner applies the strategy to all the transactions rather than a single targeted one, it will quickly use up the space of its future microblocks.
As a result, the selfish miner cannot include another transaction in its microblock, thereby losing the associated transaction fee. 
In other words, the existing analysis ignores the impact of transaction size and microblock capacity, which magnifies the selfish miner's potential revenue from the attack. 
On the other hand, the existing analysis works well for whale transactions with high fees, which are rare so that we don't need to worry about the space.
As we explained before, most of the transactions in current blockchain systems have low fees. Therefore, we need to develop a new analysis for those transactions.

\begin{figure}[!ht]
\centering
\begin{tikzpicture}[x=0.8cm,y=0.8cm]
    \node (k1) at (0 , 1) [keyblockH] {$\mathcal{B}_j$};
    \node (m1) at (1.75 , 1) [microblock] {$tx$};
    \node (k2) at (3.5 , 2.5) [keyblockA_potential] {$\mathcal{B}_{j+1}$};
    \node (m2) at (5.25 , 2.5) [microblock] {$tx$};
    \node (k3) at (7 , 2.5) [keyblockA_potential] {$\mathcal{B}_{j+2}$};
    \draw [line width=0.7pt,->,dashed, bend right=10] (k2) edge (k1);
    \draw [thick,->,>=stealth] (m1) -- (k1);
    \draw [thick,->,dashed,>=stealth] (m2) -- (k2);
    \draw [thick,->,dashed,>=stealth] (k3) -- (m2);
\end{tikzpicture}
\caption{\textbf{A simple example of the longest chain extension attack.} The selfish miner mines directly on block $B_j$ and tries to include transaction $tx$ in its future microblock.}
\label{fig:limit}
\end{figure}
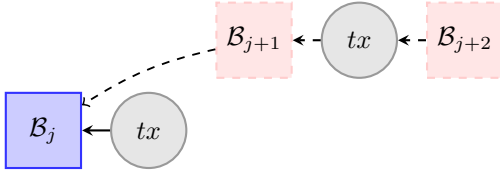

Next, we turn our attention to another limitation of the existing analysis. It assumes that
the selfish miner always adopts honest miners’ key blocks and immediately publishes its new key blocks (i.e., honest mining of key blocks). In other words, it does not consider the impact of key-block selfish mining.
This assumption can only be justified when the selfish miner's computation power is less than the threshold of making key-block selfish mining profitable, because the optimal mining strategy for key blocks is indeed honest mining~\cite{eyal2014majority,bitcoinng,sapirshtein2016optimal}.
However, once the selfish miner's computation power $\alpha$ is above the
threshold, the selfish miner has the incentive to launch the key-block mining attack and so the impact of the key-block selfish mining cannot be ignored anymore. 
This motivates us to study the selfish mining of both microblock and key block in Sec.~\ref{sec:MDP} when $\alpha$ is above the threshold.

\section{Revisiting Incentive Analysis for Microblock} \label{sec:extended}
In this section, we consider the scenario where $\alpha$ is smaller than the threshold of making key-block selfish mining profitable. This allows us to focus only on incentive analysis for microblocks. Recall that there are two types of transactions: whale transactions and regular transactions. 
As explained in Sec.~\ref{subsec:limitation}, whale transactions are so rare that they use little microblock space. For this reason, we can ignore their space requirement (even under the network capacity constraints) and apply the existing analysis.

On the other hand, regular transactions consume most of the microblock space. Hence, we can no longer ignore their space requirement and need to develop a new analysis. To this end, we consider the revenue of transaction fees in terms of regular transactions for the selfish miner and honest miners during a time interval $[0, t]$. 
Without loss of generality, we assume that there exists a block $B_0$ that the selfish miner and honest miners both agree to mine on at the starting time. (For example, $B_0$ can be the genesis block.) {Let $M(t)$ be the number of key blocks mined during the time interval $[0, t]$.} Let $X_i$ ($i \in [0, M(t)]$) denote an indicator random variable which equals one if the $i$-th key block is a selfish key block, as described below
\begin{equation}\nonumber
X_i =\left\{
\begin{aligned}
1, \ &\text{selfish key block}\\
0, \ &\text{honest key block}  .
\end{aligned}
\right.
\end{equation}
Without loss of generality, we assume block $B_0$ is an honest key block. 
For other key blocks, the possibility that it is a selfish key block is equal to $\alpha$. 

After mining a key block, its owner can issue a series of microblocks at a constant rate $v$ until the next key block is mined. Here, the rate $v$ captures the network capacity constraints. Let $Y_i$ denote the interval between the $i$-th key block and $(i+1)$-th key block. Thus, the number of produced microblocks between $i$-th and $(i+1)$-th key blocks is $vY_i$.
In addition, each microblock contains a total fee of $R_t$ because we only consider regular transactions here.
We are now ready to compute the suitable value of $r$ to resist the two microblock attacks for regular transactions.

\subsection{Resisting Transaction Inclusion Attack}
Recall the attack from Sec.~\ref{sec:mining strategy} that the selfish miner hides some of its microblocks generated after a key block but keeps mining on top of the microblock chain. Hence, honest miners directly mine on top of the selfish miner's last published block. Let $\rho$ denote the fraction of the unpublished microblocks among all the selfish microblocks between two consecutive key blocks.
In particular, $\rho=1$ means that the selfish miner hides all the microblocks it has generated between two consecutive key blocks.
Thus, if any two consecutive key blocks satisfy $(X_i, X_{i+1}) = (1,0)$ , there are $ (1 - \rho) v Y_i$ microblocks between them from the view of an honest miner; otherwise, there are $ v Y_i $ microblocks. 

Let $Z_i$ denote an indicator random variable equal to one if $\{ X_i = 1,X_{i+1} = 0 \}$, and equal to zero otherwise. Next, let $Z = \sum_{i=1}^{M(t)-1}Z_i$. Suppose $M(t)=m$.
The following lemma will aid us to bound the value of $Z$ with high probability:

\begin{lemma}\label{lem:keypair}
For $m$ consecutive key blocks, the number of block pairs $(X_i, X_{i+1}) = (1,0)$ has the following Chernoff-type bound: For $0 < \delta < 1$,
\begin{equation}
   \Pr( | Z -  \alpha \beta  (m-1)| > \delta \alpha \beta  (m-1) ) < e^{-\Omega\left(\delta^2 \alpha \beta  m\right)}. 
\end{equation}
\end{lemma}
\begin{proof}
Without loss of generality, we assume that $m$ is odd. 
Let $Z^{\textsc{odd}} = Z_1 + Z_3 + \cdots + Z_{m - 2}$ and 
$Z^{\textsc{even}} = Z_2 + Z_4 + \cdots + Z_{m - 1}$.
Then, $Z = Z^{\textsc{odd}} + Z^{\textsc{even}}$. It is easy to show that 
$E\left( Z^{\textsc{odd}} \right) = E\left( Z^{\textsc{even}} \right) = \alpha \beta (m-1)/2$, since $P\{Z_i = 1\} = 
P\{ X_i = 1\} P\{ X_{i+1} = 0 \} = \alpha \beta$.
Note that $\{ Z_{1}, Z_{3}, \ldots, Z_{m-2} \}$ are independent random variables, because $Z_{i}$ is a function of
$(X_{i},X_{i+1})$. Hence, $Z^{\textsc{odd}}$ is a sum of i.i.d. random variables. So is $Z^{\textsc{even}}$.
By Lemma~\ref{lem:key_step}, we have 
\[\Pr \left( Z < (1 - \delta) \alpha \beta  (m-1) \right) <  e^{-\Omega\left(\delta^2 \alpha \beta  (m-1) \right)} = e^{-\Omega\left(\delta^2 \alpha \beta m\right)}.
\]
Similarly, we have $\Pr \left( Z > (1 + \delta) \alpha \beta  (m-1) \right) < e^{-\Omega\left(\delta^2 \alpha \beta m\right)}$. 
\end{proof} 
This lemma shows that as $m$ increases, the number of key pairs $(X_i, X_{i+1}) = (1, 0)$ is between $(1-\delta)\alpha\beta m$ and $(1+\delta)\alpha\beta m$ with high probability. 

Next, we compute the selfish miner's relative revenue for large $m$.
On the one hand, the total amount of transaction fees for all the miners is given by $\sum_{i = 1}^{m-1}{( v Y_i R_t - \rho v  Z_i Y_i R_t)}$.
To see this, note that there are $\sum_{i = 1}^{m-1}{ v Y_i}$ microblocks produced with associated transaction fees $\sum_{i = 1}^{m-1}{ v Y_i R_t}$.
Note also that once $Z_i = 1$, there are $\rho vY_i$ microblocks not being included in the longest chain due to the transaction inclusion attack. Hence, the associated loss of transaction fees is $\sum_{i = 1}^{m-1}{\rho v  Z_i Y_i R_t}$.
On the other hand, the total transaction fees for the selfish miner is given by $\sum_{i = 1}^{m-1}{( \alpha v Y_i R_t - r \rho v  Z_i Y_i R_t})$.
To see this, note that without any attack, 
the selfish miner can get $\alpha$ fraction of the total transaction fees given by $\sum_{i = 1}^{m-1}{\alpha v Y_i R_t}$. Note also that with the transaction inclusion attack, the selfish miner will lose $r$ fraction of the total loss of transaction fees as the first leader.
Combining the above analysis,  we have the following lemma for large $m$. 
\begin{lemma}
The selfish miner's relative revenue $u$ converges to $\frac{\alpha  - r \alpha \beta \rho}{1  - \alpha \beta \rho}$ with high probability as $m \to \infty$. 
\end{lemma}

\begin{proof}
According to the previous analysis, we have 
\begin{equation} \label{eq:reveune}
\begin{split}
    u &= \lim_{m \rightarrow \infty } \frac{\sum_{i = 1}^{m-1}{( \alpha v Y_i R_t - r \rho v  Z_i Y_i R_t})}{\sum_{i = 1}^{m-1}{( v Y_i R_t - \rho v  Z_i Y_i R_t)}} \\
        &= \lim_{m \rightarrow \infty } \frac{\sum_{i = 1}^{m-1}{(\alpha v Y_i R_t - r \rho v  Z_i Y_i R_t)}/(m-1)} {\sum_{i = 1}^{m-1}{( v Y_i R_t - \rho v  Z_i Y_i R_t)} /(m-1)}\\
        &\to \frac{v R_t (\alpha  - r \alpha \beta \rho)/f}{v R_t (1  - \alpha \beta \rho)/f} \\
        &= \frac{ \alpha  - r \alpha \beta \rho}{1  - \alpha \beta \rho},
\end{split}
\end{equation}
where the third step comes from the facts that (1) $Y_i$ follows the exponential distribution with mean $1/f$ and that (2) $Y_i$ is independent of $Z_i$.
More specifically, we have $\lim_{m \rightarrow \infty} \sum_{i = 1}^{m-1}{\frac{Y_i}{m-1}} \to 1/f$ by the law of large numbers 
and $\lim_{m \rightarrow \infty }  \sum_{i = 1}^{m-1}{ \frac{Z_i Y_i}{m-1}} \to \alpha \beta /f$ by Lemma~\ref{lem:keypair}.
\end{proof}
This lemma says that for large $m$, the selfish miner's relative revenue is $\frac{(\alpha  - r \alpha \beta \rho)}{(1  - \alpha \beta \rho)}$. 
{Recall that the key block generation process is a Poisson process with rate $f$, and so $M(t)$ is a Poisson arrival process. Hence, when $t$ tends to infinity, $M(t)/t \to f$ holds with high probability. Therefore,}
with high probability, the maximum relative revenue of the selfish miner during $[0, t]$ is
\begin{equation}
\begin{aligned}
    u &= \max_{0 \leq \rho \leq 1}{\frac{\alpha - r \alpha \beta \rho}{1  - \alpha \beta \rho}} \\
    &= r + \max_{0 \leq \rho \leq 1}{\frac{\alpha - r}{1 - \alpha \beta \rho}}.
\end{aligned}
\end{equation}
If $r \leq \alpha$, the optimal $\rho = 1$ and the corresponding
\[u = r + \frac{\alpha - r}{1 - \alpha \beta}.\]
In this case, $u$ is always larger than $\alpha$ since $1 - \alpha \beta < 1$. This means that the selfish miner can always have a relative revenue greater than its fair share by utilizing this attack. On the other hand, if $r > \alpha$, the optimal $\rho = 0$ and $u = \alpha$. This means that the maximum relative revenue that the selfish miner can obtain is honest mining (i.e., $\rho = 0$). 
Therefore, we should set $r > \alpha$ in order to guarantee the adversary cannot gain more from the transaction inclusion attack.

\subsection{Resisting Longest Chain Extension Attack}
Recall the attack from Sec.~\ref{sec:mining strategy} that the selfish miner can bypass some honest microblocks and mines directly on an old honest block. Similarly, let $\rho$ denote the rejected microblock fraction. In particular, $\rho = 1$ means that the selfish miner rejects all honest microblocks and mines directly on the last honest key block. 
More precisely, if two consecutive key blocks are $(X_i, X_{i+1}) = (0,1)$, there are $ (1 - \rho) v Y_i $ honest microblocks accepted by the longest chain.
Let $K_i$ denote an indicator random variable equal to one if $\{ X_i = 0,X_{i+1} = 1\}$, and equal to zero otherwise. Let $K = \sum_{i=1}^{m-1}K_i$. The following lemma will aid us to bound the expectation of $K$ for $m$ blocks:

\begin{lemma} \label{lem:keypair2}
For the $m$ block sequence, the number of block-pair $(X_i, X_{i+1}) = (0,1)$ has the following Chernoff-type bound: For $0 < \delta < 1$,
\begin{equation}
   \Pr( | K -  \alpha \beta  (m-1)| > \delta \alpha \beta  (m-1) ) <  e^{-\Omega\left(\delta^2 \alpha \beta m \right)}. 
\end{equation}
\end{lemma}

\begin{proof}
The proof is similar to Lemma~\ref{lem:keypair}. We omit it here due to space constraints.
\end{proof} 

Next, we compute the selfish miner's relative revenue for large $m$.
On the one hand, the total amount of transaction fees for all the miners is given by $\sum_{i=1}^{m-1}{ (v Y_i R_t - \rho v K_i Y_i R_t)}$.
To see this, recall that there are $\sum_{i = 1}^{m-1}{ v Y_i}$ microblocks produced with associated transaction fees $\sum_{i = 1}^{m-1}{ v Y_i R_t}$.
Once $K_i = 1$, there are $\rho vY_i$ microblocks not being included in the longest chain due to the longest chain extension attack. Hence, the associated loss of transaction fees is $\sum_{i = 1}^{m-1}{\rho v  K_i Y_i R_t}$.
On the other hand, the total transaction fees for the selfish miner is given by $\sum_{i = 1}^{m-1}{( \alpha v Y_i R_t - r \rho v  Z_i Y_i R_t})$.
To see this, recall that without any attack, 
the selfish miner can get $\alpha$ fraction of the total transaction fees given by $\sum_{i = 1}^{m-1}{\alpha v Y_i R_t}$. 
With the longest chain extension attack, the selfish miner will lose $1-r$ fraction of the total loss of transaction fees as the second leader.
Combining the above analysis,  we have the following lemma for larger $m$
\begin{lemma}
The selfish miner's relative revenue $\mu$ converges to $\frac{ \alpha  - (1- r) \alpha \beta \rho}{1  - \alpha \beta \rho}$ with high probability as $m \to \infty$. 
\end{lemma}

\begin{proof}
According to the previous analysis, we have 
\begin{equation}\label{eq:reveune1}
\begin{split}
    u &= \lim_{m \rightarrow \infty } \frac{\sum_{i = 1}^{m-1}{( \alpha v Y_i R_t - (1 - r) \rho v  K_i Y_i R_t})}{\sum_{i = 1}^{m-1}{( v Y_i R_t - \rho v  K_i Y_i R_t)}} \\
        &= \lim_{m \rightarrow \infty } \frac{\sum_{i = 1}^{m-1}{(\alpha v Y_i R_t - (1-r) \rho v  K_i Y_i R_t)}/(m-1)} {\sum_{i = 1}^{m-1}{( v Y_i R_t - \rho v  K_i Y_i R_t)} /(m-1)}\\
        &\to \frac{v R_t (\alpha  - (1-r) \alpha \beta \rho)/f}{v R_t (1  - \alpha \beta \rho)/f} \\
        &= \frac{ \alpha  - (1- r) \alpha \beta \rho}{1  - \alpha \beta \rho},
\end{split}
\end{equation}
where the third step comes from the facts that $Y_i$ follows the exponential distribution with mean $1/f$ and that $Y_i$ is independent of $K_i$.
\end{proof}

This lemma says that for large $m$, the selfish miner's relative revenue is $\frac{ \alpha  - (1-r) \alpha \beta \rho}{1  - \alpha \beta \rho}$. {Similar with the previous analysis, we can show that as $t \to \infty$,}
with high probability, the maximum relative revenue of the selfish miner during $[0, t]$ is
\begin{equation}
\begin{aligned}
    u   &= \max_{0 \leq \rho \leq 1}{\frac{\alpha-(1-r)\alpha \beta \rho}{ 1 - \alpha \beta \rho}} \\
        &= 1 - r + \max_{0 \leq \rho \leq 1}{\frac{r - \beta}{1 - \alpha \beta \rho}}.
\end{aligned}
\end{equation}
If $r \geq \beta$, the optimal $\rho = 1$ and the corresponding
\[u = 1 - r + \frac{r - \beta}{1 - \alpha \beta}.\]
In this case, $u$ is always larger than $\alpha$ since $1-\alpha \beta < 1$. This means that the selfish miner can always have a relative revenue greater than its fair share by launching this attack.
On the other hand, if $r < \beta$, the optimal $\rho = 0$ and $u = \alpha$. This means that the maximum relative revenue that the selfish miner can obtain is honest mining (i.e., $\rho = 0$).
Therefore, we should set $r < \beta$ in order to guarantee the adversary cannot gain more from the longest chain extension attack.
Combining the two incentive sub-mechanisms of transaction inclusion and longest chain extension, the value of $r$ needs to satisfy that \[\alpha < r < \beta.\]

\subsection{Discussion}
We show that the existing analysis in \cite{bitcoinng, yin18bngrrevisit} can be used to bound the split ratio $r$ for whale transactions. Also, we develop new analysis to bound the ratio $r$ for regular transactions under network capacity constraints. These bounds are depicted in Fig.~\ref{fig:ratiocompare}.
As we can see, our bound $\alpha< r < \beta$ (for regular transactions) contains the previous two bounds $1-\frac{1-\alpha}{1+\alpha-\alpha^2}<r<\frac{1-\alpha}{2-\alpha}$
and $\frac{\alpha}{1-\alpha}<r<\frac{1-\alpha}{2-\alpha}$ (for whale transactions). This leads to several interesting implications.

\begin{figure}[!ht]
\centering
\includegraphics[width=1.8in,height=1.6in]{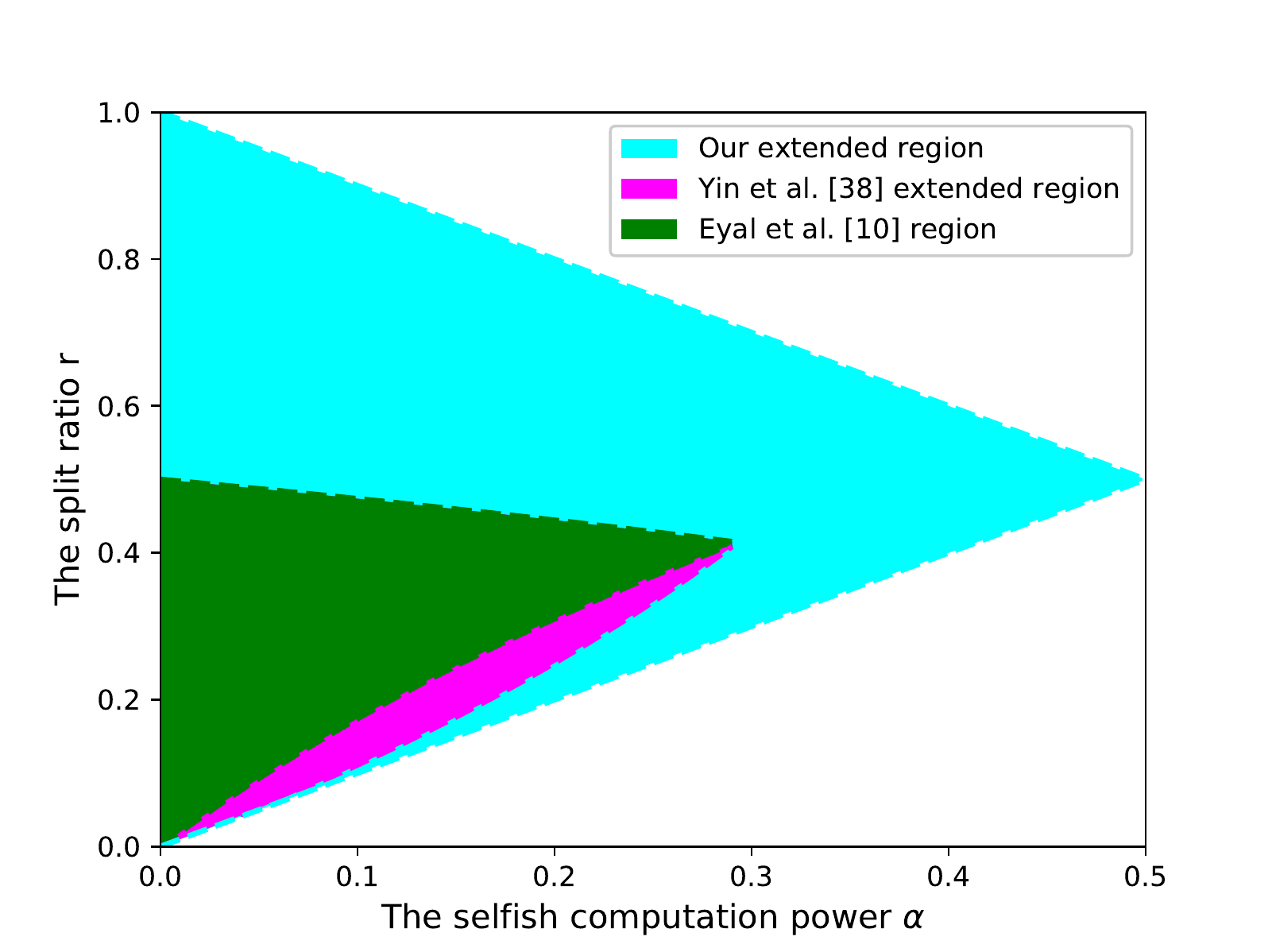}
\caption{\textbf{The comparison of the transaction fee distribution ratio.}}
\vspace{-2mm}
\label{fig:ratiocompare}
\end{figure}

First, introducing network capacity constraints doesn't make it harder to maintain the incentive compatibility of Bitcoin-NG. This is because the bounds for whale transactions are the same as the previous ones and the bound for regular transactions contains the previous ones.

Second, when $\alpha$ is smaller than $29\%$, we can find a value of $r$ that satisfies all the bounds. This means that the incentive compatibility of Bitcoin-NG can be maintained for all types of transactions even under network capacity constraints in this regime.

Third, when $\alpha$ is larger than $29\%$, we cannot find a value of $r$ that satisfies all the bounds, because the two bounds for whale transactions both become invalid. This means that the incentive compatibility of Bitcoin-NG can be maintained only for regular transactions but not for whale transactions in this regime. In other words, the presence of whale transactions might cause instability of the whole system in this regime. As such, some defense mechanisms should be designed accordingly.

\begin{remark}
{As shown in Eq.~\eqref{eq:reveune} and Eq.~\eqref{eq:reveune1}, as the network capacity $v$ increases, the revenues for both the selfish miner and honest miners also increase. In particular, the expected revenues for them are proportional with the network capacity. This is because our analysis assumes that the pending transactions always exist and so the change of the network capacity does not affect the attack strategies. 
More importantly, as both the revenues for the selfish miner and honest miners are increasing proportionally with the network capacity, the final relative revenue doesn't depend on the network capacity $v$.}
\end{remark}

\section{Incentive Analysis for both Microblock and Key Block} \label{sec:MDP}
In this section, we consider the scenario where $\alpha$ is larger than the threshold of making key-block selfish mining profitable. This requires us to study the selfish mining strategies for key blocks and microblocks jointly. To the best of our knowledge, we are among the first to conduct such an analysis, since most previous analysis only focuses on the incentive of microblocks.

\subsection{Markov Decision Process (MDP)}
We apply MDP to model various selfish mining strategies. Although MDP has been used to model the selfish mining strategies in Bitcoin (see, e.g., \cite{sapirshtein2016optimal}), its application to Bitcoin-NG is new and non-trivial due to the presence of the transaction inclusion attack and the longest chain extension attack.  

In order to make our analysis tractable, we introduce two simplifications. First, we assume that the key block interval is $1/f$ and so the number of microblocks produced between two consecutive key blocks is $v/f$. (Note that the key block interval is assumed to follow the exponential distribution with mean $1/f$ in Sec.~\ref{sec:extended}.) Second, we only consider binary choices: publishing or hiding all selfish microblocks in the transaction inclusion attack, and accepting or rejecting all honest microblocks in the longest chain extension attack. (This is consistent with the fact that $\rho = 0$ or $1$ in Sec.~\ref{sec:extended}.)

We are now ready to describe our MDP model, which can be presented by a $4$-tuple $\mathcal{M}:=(S,A,P,R)$, where $S$ is the state space, $A$ is the action space, $P$ is the stochastic state transition matrix, and $R$ is the reward matrix. Specifically, $S$ contains all possible states in the selfish mining process; $A$ includes the available actions (e.g., publishing or hiding blocks by the selfish miner) at each state;  $P$ contains the transition probabilities from the current state to the next state according to the taken action; $R$ records how much the selfish miner obtains when there are some state transitions.
Table~\ref{tab:state} illustrates the MDP of selfish mining in Bitcoin-NG. 
Note that blocks are assumed to be transmitted without delay (see Sec.~\ref{sec:model}), and so forks are not considered in the analysis.
Below we will discuss each component of the $4$-tuple:

\noindent \textbf{Actions.} The selfish miner has eight available actions in our model.  
\begin{itemize}[leftmargin=*]
    \item \textbf{Adopt and include.} The selfish miner accepts all honest key blocks and the corresponding honest microblocks. In other words, the selfish miner will mine its key block on the last honest block and abandon its private chain. This action is referred to as $\mathsf{adopt}$. 
    
    \item \textbf{Adopt and exclude.} The selfish miner accepts all honest key blocks and microblocks except for microblocks produced after the last honest key block. Specifically, the selfish miner directly mines on top of the last honest key block, which is referred to as $\mathsf{adoptE}$.
    
    \item \textbf{Override and publish.} The selfish miner publishes all its key blocks and corresponding microblocks whenever its private chain is longer than the honest one. The chain length is counted by the key block. This action is denoted as $\mathsf{override}$. 
    
    \item \textbf{Override and hide.} The selfish miner publishes all its key blocks and the microblocks except for these mined after the last selfish key block whenever its private chain is longer than the honest one. 
    This action is denoted as $\mathsf{overrideH}$. 
    
    \item \textbf{Match and publish.} When an honest miner finds one new key block, the selfish miner publishes its key block of the same height and the microblocks built after this key block. This action is available when the selfish miner has one block in advance and is referred to as $\mathsf{match}$.
    
    \item \textbf{Match and hide.} When an honest miner generates a new key block, the selfish miner publishes its key block of the same height while hides the microblocks built after this key block. This action is also available when the selfish miner has one block in advance. This action is denoted as $\mathsf{matchH}$.
    
    \item \textbf{Wait.} In this action, the selfish miner does not publish any new key blocks and microblocks, while keeps mining on its private chain until a new key block and corresponding microblocks are found.
    
    \item \textbf{Revert.} The selfish miner reverts its previous actions. Specifically, the selfish miner can publish its hidden microblocks when there is no honest key block mined after its block; the selfish miner can include the honest microblocks (decided to excluded in the previous decision) or excludes the honest microblocks (decided to included in the previous decision) once there is no selfish key block mined on an honest block.
\end{itemize}
{The adopt, override match, and wait actions include all possible actions on the selfish mining of key blocks, while hide, publish and revert actions cover all possible actions on the transaction inclusion and longest chain extension attacks of microblocks.}
Note that in match action, the selfish miner publishes its key block of the same height to match honest miners' key block. Therefore, there are two forking branches of the same length.
In Bitcoin-NG, honest miners adopt a uniform tie breaking rule to choose which branch to mine on. 
See Sec.~\ref{subsec:bitcoinng} for more details about the fork choice rule in Bitcoin-NG.
In particular, we introduce a variable $\gamma$, which denotes the fraction of honest miners that mine on the selfish miner's branch. 
\begin{table*}[!ht]
    \caption{State transition and reward matrices for the optimal selfish mining.}
    \label{tab:state}
    \centering
    \begin{tabular}{c|c|c|c|c}
    \hline
    \textbf{State $\times$ Action} & \textbf{State} & \textbf{Probability} & \textbf{Reward} & \textbf{Condition}\\
    \hline
    $(l_a,l_h,\cdot, S_h),$ $\mathsf{adopt}$ & \multirow{3}*{\shortstack{$(1, 0,\mathsf{noTie}, H_\mathsf{in})$ \\ $(0, 1,\mathsf{noTie}, H_\mathsf{in})$ }}  & \multirow{3}*{\shortstack{$\alpha$ \vspace{0.15cm} \\ $1 - \alpha$ }} & $(l_h, l_h, 0, 0)$ & \multirow{3}*{$-$} \\ \cline{1-1} \cline{4-4}
    $(l_a,l_h,\cdot, S_p),$ $\mathsf{adopt}$ & & & $(l_h, l_h-1+(1-r), 0, r)$ & \\ \cline{1-1} \cline{4-4}
    $(l_a,l_h,\cdot, \left\{H_\mathsf{in}, H_\mathsf{ex}\right\}),$ $\mathsf{adopt}$ & & & $(l_h, l_h-1, 0, 0)$ &\\ \cline{1-1} \cline{4-4}
    
    \hline
    $(l_a,l_h,\cdot, S_h) ,$ $\mathsf{adoptE}$ & \multirow{3}*{\shortstack{$(1, 0,\mathsf{noTie}, H_\mathsf{ex})$  \\ $(0, 1,\mathsf{noTie}, H_\mathsf{ex})$ }}  & \multirow{3}*{\shortstack{$\alpha$ \vspace{0.15cm}  \\ $1 - \alpha$ }} & $(l_h, l_h, 0, 0)$ & \multirow{3}*{$-$} \\ \cline{1-1} \cline{4-4}
    $(l_a,l_h,\cdot, S_p) ,$ $\mathsf{adoptE}$ & & & $(l_h, l_h-1+(1-r), 0, r)$  & \\ \cline{1-1} \cline{4-4}
    $(l_a,l_h,\cdot, \left\{H_\mathsf{in}, H_\mathsf{ex}\right\}),$ $\mathsf{adoptE}$ & & & $(l_h, l_h-1, 0, 0)$ & \\ \cline{1-1} \cline{4-4}
    
    \hline
    $(l_a,l_h,\cdot, H_\mathsf{ex}),$ $\mathsf{override}$  & \multirow{3}*{\shortstack{$(l_a-l_h, 0,\mathsf{noTie}, S_p)$  \\ $(l_a-l_h-1, 1,\mathsf{noTie}, S_p)$ }}  & \multirow{3}*{\shortstack{$\alpha$ \vspace{0.15cm} \\ $1 - \alpha$ }} & $(0, 0, l_h+1, l_h+1)$ & \multirow{3}*{$l_a > l_h$} \\ \cline{1-1} \cline{4-4}
    $(l_a,l_h,\cdot, H_\mathsf{in}),$ $\mathsf{override}$ & & & $(0, r, l_h+1, l_h+(1-r))$  & \\ \cline{1-1} \cline{4-4}
    $(l_a,l_h,\cdot, \left\{S_p, S_h\right\}),$ $\mathsf{override}$ & & & $(0, 0, l_h+1, l_h)$ & \\ \cline{1-1} \cline{4-4}
    
    \hline
    $(l_a,l_h,\cdot, H_\mathsf{ex}),$ $\mathsf{overrideH}$ & \multirow{3}*{\shortstack{$(l_a-l_h, 0,\mathsf{noTie}, S_h)$  \\ $(l_a-l_h-1, 1,\mathsf{noTie}, S_h)$ }}  & \multirow{3}*{\shortstack{$\alpha$ \vspace{0.15cm} \\ $1 - \alpha$ }} & $(0, 0, l_h+1, l_h+1)$ & \multirow{3}*{$l_a > l_h$} \\ \cline{1-1} \cline{4-4}
    $(l_a,l_h,\cdot, H_\mathsf{in}),$ $\mathsf{overrideH}$ & & & $(0, r, l_h+1, l_h+(1-r))$  & \\ \cline{1-1} \cline{4-4}
    $(l_a,l_h,\cdot, \left\{S_p, S_h\right\}),$ $\mathsf{overrideH}$ & & & $(0, 0, l_h+1, l_h)$ & \\ \cline{1-1} \cline{4-4}
    
    \hline
    \multirow{2}{*}{$(l_a,l_h,\mathsf{noTie}, \cdot),$ $\mathsf{wait}$}    & $(l_a+1, l_h,\mathsf{noTie}, *)$ & $\alpha$  & \multirow{2}{*}{$(0, 0, 0, 0)$} & \multirow{2}*{$-$}\\
                                                                                & $(l_a, l_h+1,\mathsf{noTie}, *)$ & $1 - \alpha$ &  \\  
    \hline
    \multirow{3}*{\shortstack{$(l_a,l_h,noTie, H_\mathsf{in}),$ $\mathsf{match}$ \\ $(l_a,l_h,\mathsf{tie}, H_\mathsf{in}),$ $\mathsf{wait}$}}  
                                                         & $(l_a+1, l_h, \mathsf{tie}, H_\mathsf{in})$ & $\alpha$  & $(0, 0, 0, 0)$ & \multirow{3}*{$la \geq l_h$}\\
                                                         & $(l_a-l_h, 1,\mathsf{noTie}, S_p)$ & $\gamma (1- \alpha)$  & $(0, r, l_h, l_h-1+(1-r))$ &\\
                                                         & $(l_a, l_h+1,\mathsf{noTie}, H_\mathsf{in})$ & $(1-\gamma)(1- \alpha)$  & $(0, 0, 0, 0)$ &\\
    \hline
    \multirow{3}*{\shortstack{$(l_a,l_h,noTie, H_\mathsf{ex}),$ $\mathsf{match}$ \\ $(l_a,l_h,\mathsf{tie}, H_\mathsf{ex}),$ $\mathsf{wait}$}} & $(l_a+1, l_h, \mathsf{tie}, H_\mathsf{ex})$ & $\alpha$  & $(0, 0, 0, 0)$ & \multirow{3}*{$la \geq l_h$}\\
                                                         & $(l_a-l_h, 1,\mathsf{noTie}, S_p)$ & $\gamma (1- \alpha)$  & $(0, 0, l_h, l_h-1)$ & \\
                                                         & $(l_a, l_h+1,\mathsf{noTie}, H_\mathsf{ex})$ & $(1-\gamma)(1- \alpha)$  & $(0, 0, 0, 0)$ & \\
    \hline
    \multirow{3}{*}{\shortstack{$(l_a,l_h,\mathsf{noTie}, \left\{S_p, S_h\right\}),$ $\mathsf{match}$ \\ $(l_a,l_h,\mathsf{tie}, \left\{S_p, S_h\right\}),$ $\mathsf{wait}$}}   
                                                         & $(l_a+1, l_h, \mathsf{tie}, *)$ & $\alpha$  & $(0, 0, 0, 0)$ & \multirow{3}*{$la \geq l_h$} \\
                                                         & $(l_a-l_h, 1,\mathsf{noTie}, S_p)$ & $\gamma (1- \alpha)$  & $(0, 0, l_h, l_h)$ & \\
                                                         & $(l_a, l_h+1,\mathsf{noTie}, *)$ & $(1-\gamma)(1- \alpha)$  & $(0, 0, 0, 0)$ & \\
    \hline
    \multirow{3}{*}{\shortstack{$(l_a,l_h,\mathsf{noTie}, H_\mathsf{in}),$ $\mathsf{matchH}$ \\ $(l_a,l_h, tie^{\prime}, H_\mathsf{in}),$ $\mathsf{wait}$}}    
                                                         & $(l_a+1, l_h, tie^{\prime}, H_\mathsf{in})$ & $\alpha$  & $(0, 0, 0, 0)$ & \multirow{3}*{$la \geq l_h$}\\
                                                         & $(l_a-l_h, 1,\mathsf{noTie}, S_h)$ & $\gamma (1- \alpha)$  & $(0, r, l_h, l_h-1+(1-r))$ &\\
                                                         & $(l_a, l_h+1,\mathsf{noTie}, H_\mathsf{in})$ & $(1-\gamma)(1- \alpha)$  & $(0, 0, 0, 0)$ &\\
    \hline
    \multirow{3}{*}{\shortstack{$(l_a,l_h,noTie, H_\mathsf{ex}),$ $\mathsf{matchH}$ \\ $(l_a,l_h,\mathsf{tie}^{\prime}, H_\mathsf{ex}),$ $\mathsf{wait}$}}    
                                                         & $(l_a+1, l_h, \mathsf{tie}^{\prime}, H_\mathsf{ex})$ & $\alpha$  & $(0, 0, 0, 0)$ & \multirow{3}*{$l_a \geq l_h$}\\
                                                         & $(l_a-l_h, 1,\mathsf{noTie}, S_h)$ & $\gamma (1- \alpha)$  & $(0, 0, l_h, l_h-1)$ & \\
                                                         & $(l_a, l_h+1,\mathsf{noTie}, H_\mathsf{ex})$ & $(1-\gamma)(1- \alpha)$  & $(0, 0, 0, 0)$ &\\
    \hline
    \multirow{3}{*}{\shortstack{$(l_a,l_h,noTie,\left\{S_p, S_h\right\}),$ $\mathsf{matchH}$ \\ $(l_a,l_h,\mathsf{tie}^{\prime}, \left\{S_p, S_h\right\}),$ $\mathsf{wait}$}}                                                       & $(l_a+1, l_h, \mathsf{tie}^{\prime}, *)$ & $\alpha$  & $(0, 0, 0, 0)$ & \multirow{3}*{$l_a \geq l_h$}\\
                                                         & $(l_a-l_h, 1,\mathsf{noTie}, S_h)$ & $\gamma (1- \alpha)$  & $(0, 0, l_h, l_h)$ & \\
                                                         & $(l_a, l_h+1,\mathsf{noTie}, *)$ & $(1-\gamma)(1- \alpha)$  & $(0, 0, 0, 0)$ &\\    
    \hline
    $(l_a,l_h, \mathsf{tie}^{\prime}, \cdot),$ $\mathsf{revert}$ & $(l_a, l_h, \mathsf{tie}, *)$ & $1$  & $(0, 0, 0, 0)$ & $-$\\
    \hline
    $(l_a,l_h, \cdot, S_h),$ $\mathsf{revert}$ & $(l_a, l_h, *, S_p)$ & $1$  & $(0, 0, 0, 0)$ & $l_h = 0$\\                                                 
    \hline
    $(l_a,l_h, \cdot, H_\mathsf{ex}),$ $\mathsf{revert}$ & $(l_a, l_h, *, H_\mathsf{in})$ & $1$  & $(0, 0, 0, 0)$ & $l_a = 0$\\                          
    
    \hline
\end{tabular}
    \begin{tablenotes}
    \item $*$ denotes the state element remains the same in the state transition.
    \end{tablenotes}
\end{table*}

\noindent \textbf{State space.}  The state space $S$ is also composed by $4$-tuple $(l_a,l_h,\mathsf{fork}, \mathsf{lastMicroBlock})$. 
\begin{itemize}[leftmargin=*]
    \item $\bm{l_a}$ accounts for the length of the chain mined by the selfish miner after the last common ancestor key block.
    More precisely, the last common ancestor key block is the last key block in the longest chain accepted by both the selfish miner and \emph{all} honest miners, and is updated once the selfish miner adopts the public chain or all honest miners adopt the selfish miner's chain.
    In addition, the chain length is counted by the selfish key blocks in this branch. 
    
    \item $\bm{l_h}$ is the length of the public chain after the last common ancestor key block. This chain can be viewed by both the selfish miner and honest miners.

    \item $\bm{\mathsf{fork}}$. The field fork obtains three possible values, dubbed $\mathsf{noTie}$, $\mathsf{tie}$ and $\mathsf{tie}^{\prime}$. Specifically, $\mathsf{tie}$ means the selfish miner publishes $l_h$ selfish key block and the corresponding microblocks; $\mathsf{tie}^{\prime}$ presents the selfish miner publishes $l_h$ selfish key block and the corresponding microblocks except for these after the last selfish key block; $\mathsf{noTie}$ signifies that there are not two public branches with the equivalent length. 
    
    \item $\bm{\mathsf{lastMicroBlock}}$. This field also includes four possible values, dubbed $H_\mathsf{in}$, $H_\mathsf{ex}$, $S_p$, and $S_h$. Specifically, $H_\mathsf{in}$ (respectively, $H_\mathsf{ex}$) represents the common ancestor is an honest key block, and the corresponding microblocks are accepted (respectively, rejected) by the selfish miner. While $S_p$ (or $S_h$) which stands for the common ancestor is a selfish key block, and the corresponding microblocks mined are published (or hidden, respectively) by the selfish miner.
\end{itemize}

\noindent \textbf{State Transition and Reward.} We use a $4$-tuple $(R_h,T_h ,R_a$,$T_a)$ to indicate the rewards won by the selfish and honest miners in the state transitions. Specifically, $R_h$ (respectively, $R_a$) is the key block rewards for honest (respectively, the selfish) miners , while $T_h$ (respectively, $T_a$) is the transaction fee for honest (respectively, the selfish) miners. 

Recall that there are two types of transactions. Here, we focus on regular transactions and will discuss whale transactions later. Recall also that the microblock fee of a regular microblock is denoted by $R_t$. 
For convenience, instead of recording the number of rewards, each field only records the number of key block reward or transaction fees (the total transaction fee in $v/f$ microblocks as one unit) won by miners. 
More importantly, the transaction fees included in the microblocks after the common ancestor key block are not assigned to miners until the next ancestor key block is decided. This is because these transaction fees are affected by some future actions of the selfish miner (see Sec.~\ref{sec:extended}). 

In $\mathsf{adopt}$ or $\mathsf{adoptE}$ actions, the selfish miner accepts $l_h$ honest key blocks and the microblocks mined before these key blocks. Honest miners obtain $l_h R_b$ key block rewards and $(l_h - 1) v/f R_t$ transaction fees. In $\mathsf{override}$ or $\mathsf{overrideH}$ actions, the selfish miner publishes $l_h + 1$ selfish key blocks. Honest miners accept these key blocks and the microblocks produced before the key blocks. Thus, the selfish miner obtains $(l_h+1)R_b$ key block rewards and $l_h v/f R_t$ transaction fees. In the $\mathsf{match}$ actions, the next state depends on whether the next key block is created by the selfish miner (w.p. $\alpha$), by some honest miners working on the honest branch (w.p. $(1-\gamma)(1-\alpha)$), or by the left honest miners mining on the selfish branch (w.p. $ \gamma(1-\alpha)$). In the latter case, the selfish miner effectively overrides the honest miners' branch. It can obtain $l_h R_b$ key block reward and $(l_h-1) v/f R_t$ transaction fees. Note that the value of $\gamma$ is decided by the adopted fork solution (e.g., $\gamma=0.5$ in the uniform tie-break policy).

Once the common ancestor key block changed, the transaction fees in the microblocks produced after the previous ancestor key block will be assigned. There are two cases.
\begin{itemize}[leftmargin=*]
    \item The previous common ancestor key block is mined by an honest miner. This case can be further divided into two subcases: $1$) the next key block is mined by honest miners, and honest miners get $v/f R_t$ transaction fees; $2$) the next key block is mined by the selfish miner and $\bm{\mathsf{lastMicroBlock}} = H_\mathsf{in}$, honest miners get $r v/f R_t$ transaction fees and the selfish miner gets $(1-r) v/f R_t$ transaction fees. 
    
    \item The previous common ancestor key block is mined by the selfish miner. This case can be further divided into two subcases: $1$) the next key block is mined by the selfish miner, and the selfish miner gets $v/f R_t$ transaction fees; $2$) the next key block is mined by some honest miners and $\bm{\mathsf{lastMicroBlock}} = S_\mathsf{p}$, the selfish miner gets $r v/f R_t$ transaction fees and honest miners get $(1-r) v/f R_t$ transaction fees. 
\end{itemize}

Finally, we turn our attention from regular transactions to whale transactions. Since whale transactions are rare and unpredictable, we  model the microblock fee as a random variable taking two values: $R_t$ or $R_t$ plus the fee of a whale transaction. Let $\bar{R}_t$ be the expected microblock fee. Clearly, $\bar{R}_t > R_t$. The long-term effect of whale transactions is to decrease the ratio $k$ from $R_b/R_t$ to $R_b/\bar{R}_t$. As we will see in our evaluation shortly, such an effect slightly increases the relative revenue of the selfish miner.

\subsection{Evaluation Results} \label{sec:evaluation}
We use the MDP toolbox developed in MATLAB~\cite{chades2014mdptoolbox} to obtain the selfish miner's relative revenue, denoted in the equation \eqref{eq:revenue}.\footnote{{With the MDP toolbox, we can numerically obtain the optimal policies under each scenarios. We do not provide them due to space constraints.}} Note that in the following evaluations, Bitcoin-NG adopts the uniform tie-breaking policy ($\gamma = 0.5$). 

\noindent \textbf{The selfish mining threshold.} Fig.~\ref{fig:alpha} shows the relative revenues of the selfish miner when $r=0.4$ (used in Bitcoin-NG~\cite{bitcoinng}) and $\alpha\in[0,0.45]$.
We consider three reward settings: $k \rightarrow 0$, $k = v/f$, and $k \rightarrow \infty$. Specifically, in the first setting, the transaction fees dominate the miners' revenue; in the second setting, the transaction fees included in $v/f$ microblocks between two consecutive key blocks have the same weight with one key block reward; in the third setting, the key block rewards dominate miners' revenue. Note that the key block reward dominated case has a similar reward distribution as Bitcoin, i.e., the microblock architecture does not impact the system. 

\begin{figure}[!ht]
\centering
\setlength{\abovecaptionskip}{10pt}   
\setlength{\belowcaptionskip}{10pt}   
\includegraphics[width=1.8in,height=1.5in]{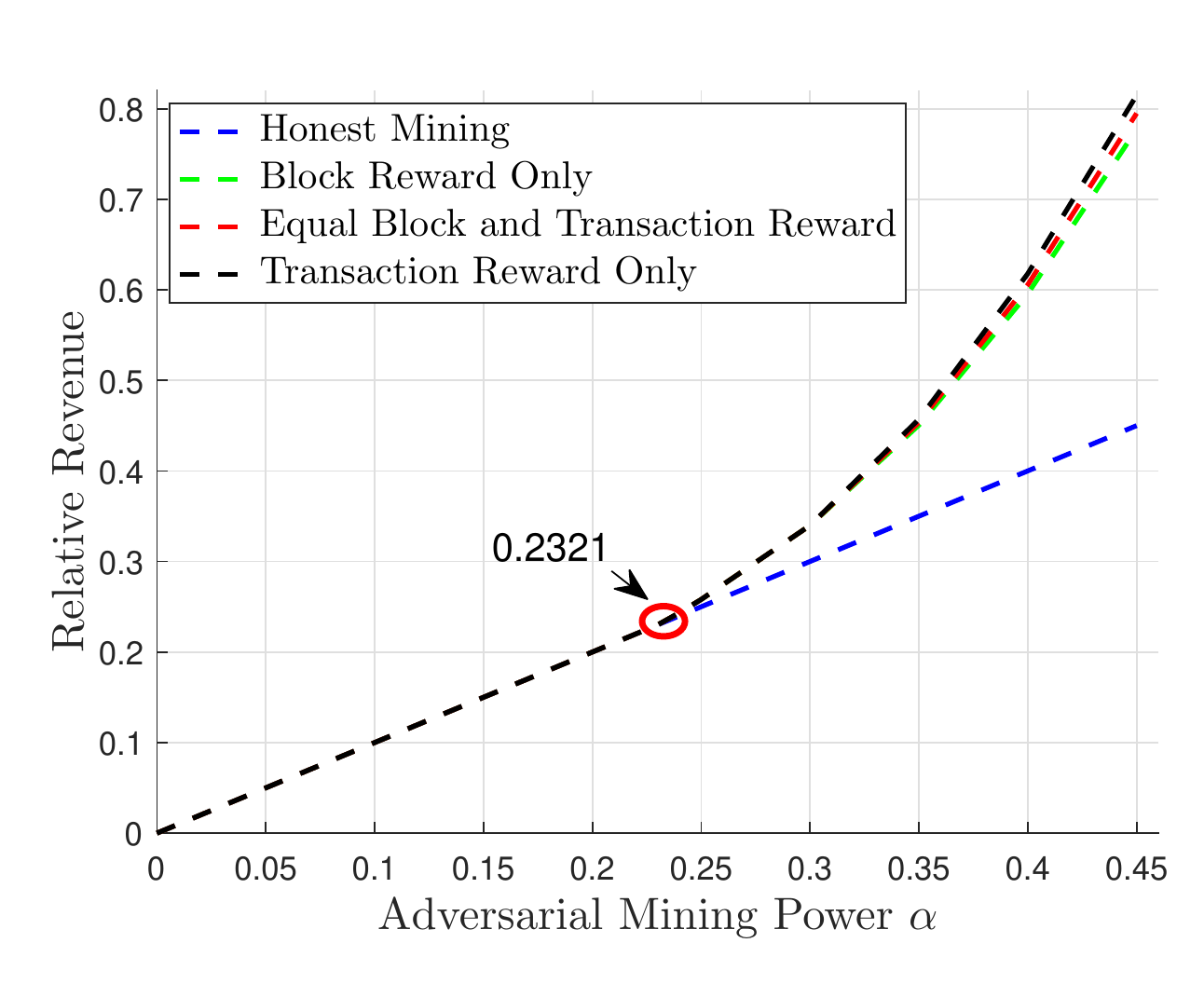}
\caption{The selfish miner's relative revenue.}
\vspace{-2mm}
\label{fig:alpha}
\end{figure}

The figure shows that the thresholds of making selfish mining profitable in these three settings are all $23.21\%$, which is the same as the selfish mining threshold in Bitcoin.
In other words, by adopting the suitable $r$ (i.e., $\alpha < r < 1 - \alpha$), the microblock architecture in Bitcoin-NG does not affect the system security compared with Bitcoin. 
In addition, the selfish miner's revenues in the three settings are still the same even when $\alpha>29\%$, which verifies our analysis in Sec.~\ref{sec:extended} and supports that Bitcoin-NG is as resilient as Bitcoin. 

When $\alpha>35\%$, the differences between the selfish miner's revenues in the three settings and the honest revenue are exhibited in Fig.~\ref{fig:alpha1}. It's easy to see that the selfish miner can obtain the highest revenue in the transaction fee dominated case. This implies that the microblock architecture can slightly increase the selfish miner's revenue, but the increase is much smaller than the results~\cite{ziyu2019}. We will explain the reason shortly. 

\begin{figure}[!ht]
\centering
\setlength{\abovecaptionskip}{10pt}   
\setlength{\belowcaptionskip}{10pt}   
\includegraphics[width=2in,height=1.5in]{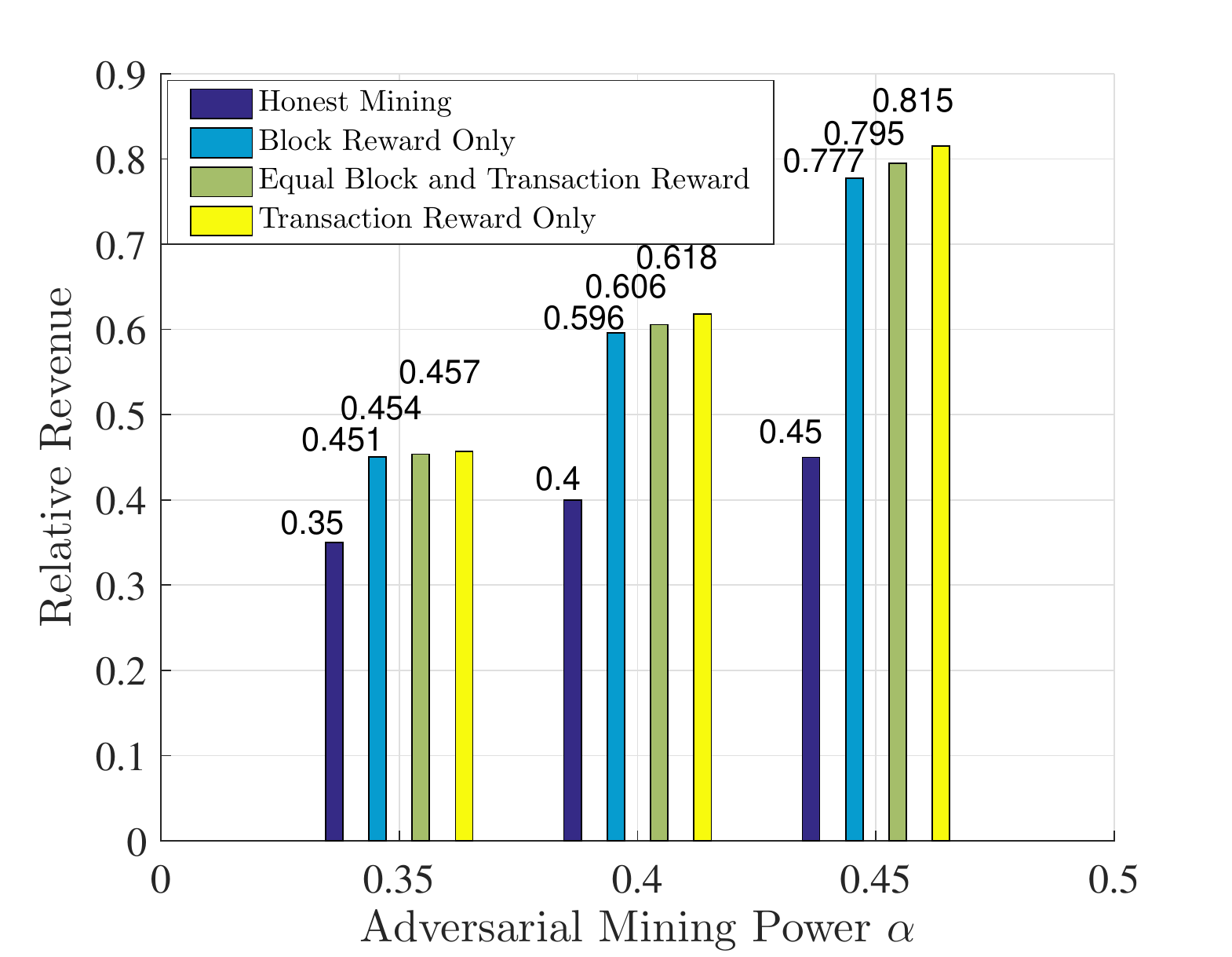}
\caption{The relative revenue when $\alpha \in \{0.35, 0.4, 0.45\}$.}
\vspace{-2mm}
\label{fig:alpha1}
\end{figure}

\noindent \textbf{The selfish revenue with different split ratio.} Fig.~\ref{fig:ratio1} shows the selfish miner's relative revenue when $\alpha = 23.21\%$ and $r$ is set to different values. In Fig.~\ref{fig:ratio1}, we consider two settings: $k \rightarrow 0$ and $k = v/f$ (which are introduced above). We can see when $r$ ranges from $0.2321$ to $0.7679$, the selfish revenues in the two settings are both the lowest. 
Therefore, the results support our findings in Sec.~\ref{sec:extended}. This is when adopting the suitable $r$ (i.e., $\alpha < r < 1-\alpha$), the selfish miner cannot gain more from microblock selfish mining.  

\begin{figure}[!ht]
\centering
\setlength{\abovecaptionskip}{10pt}   
\setlength{\belowcaptionskip}{10pt}   
\includegraphics[width=1.8in,height=1.3in]{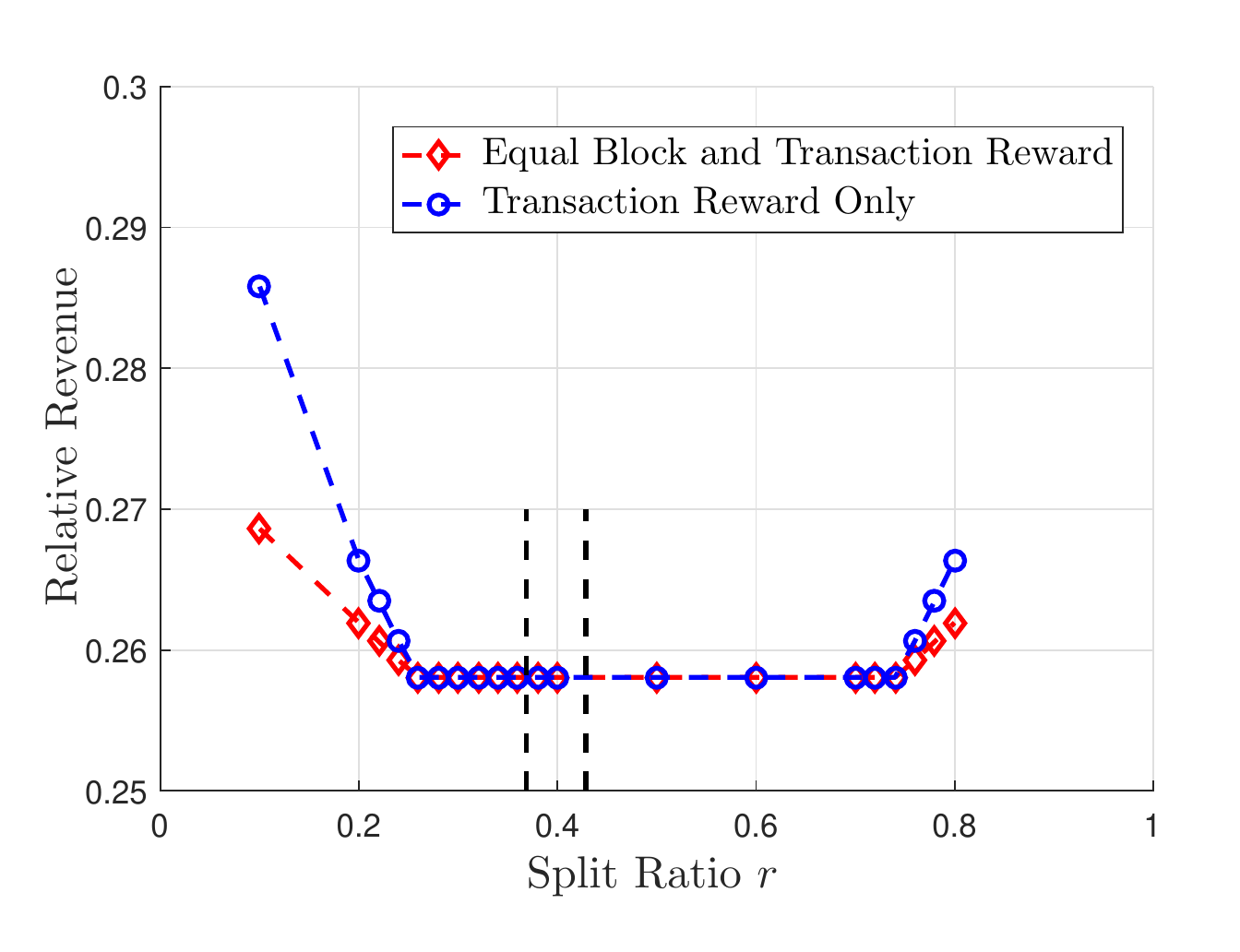}
\caption{The selfish miner's relative revenue with different split ratio $r$.}
\vspace{-2mm}
\label{fig:ratio1}
\end{figure}

\subsection{Discussion} \label{sec:comp2}
First, we validate that when adopting the suitable $r$, the selfish mining threshold in Bitcoin-NG is the same as the threshold in Bitcoin~\cite{sapirshtein2016optimal}. This means that Bitcoin-NG and Bitcoin have the same security level. In other words, our result supports the security claim of Bitcoin-NG~\cite{bitcoinng} by using MDP. 

Second, we find out that when $\alpha>35\%$, the selfish miner in Bitcoin-NG can gain more revenue than in Bitcoin. However, the increase is much less significant as shown in the previous result~\cite{ziyu2019}. This is because Wang et al.~\cite{ziyu2019} treated the selfish mining of key blocks and microblocks as two independent issues and take the outputs of key-block mining attack (i.e., the fraction of selfish key blocks) as input to the microblock mining analysis. 
Their method amplifies the selfish miner's revenue by serializing the key block and microblocks mining. By contrast, our MDP approach integrates the selfish mining of key blocks and microblocks in a more accurate way, leading to theoretical results closer to reality.

\section{Related Work} \label{sec:relatedwork}
In this section, we introduce the prior works of incentive analysis in Bitcoin and Bitcoin-NG. 

\noindent \textbf{Bitcoin Incentive Analysis.} Eyal and Sirer~\cite{eyal2014majority} show that the Bitcoin mining protocol is not incentive competitive. They also introduce a deviant strategy named selfish mining, which wastes honest power and decreases the security threshold to $0.25$. Nayak et al.~\cite{nayak2016stubborn} conclude the selfish mining strategy and extend it to the stubborn mining strategies, which also combines an Eclipse attack.
Moreover, Sapirshtein et al.~\cite{sapirshtein2016optimal} and Gervais et al.~\cite{gervais2016security} try to figure out the optimal Bitcoin selfish mining threshold utilizing the MDP tool. Carlsten et al.~\cite{instability} focus more on the deviating strategy in the transaction-fee regime where the block reward dwindles to a negligible amount. Their undercutting attack works even an attacker only accounts for small computation power and a poor network connection. After confirming the postulate of Carlsten et al.~\cite{instability}, Tsabary and Eyal~\cite{TE18} additionally study the Bitcoin gap game between block reward and transaction fee. 

\noindent \textbf{Bitcoin-NG Incentive Analysis.}
Yin et al.~\cite{yin18bngrrevisit} extended the transaction fee distribution ratio after considering another situation that the original paper omits~\cite{bitcoinng}. Wang at al.~\cite{ziyu2019} considered advanced selfish mining strategies, i.e., stubborn mining strategies, when an attacker may manipulate the microblock chains between two honest parties. However, these prior works have several limitations in the incentive analysis, as explained in  Sec.~\ref{sec:introduction}. 

\section{Conclusion}\label{sec:conclusion}
In this paper, we have proposed a new incentive analysis of Bitcoin-NG considering the network capacity. Our model enables us to evaluate the impact of key-block generation interval and microblock generation rate, which is missing in the previous analysis. In particular, we have shown that Bitcoin-NG can still maintain incentive compatibility against the microblock mining attack even under network capacity constraints. In addition, we have modeled the selfish mining of key blocks and microblocks jointly into an MDP and shown the threshold of Bitcoin-NG is a little lower than in Bitcoin only when the selfish mining power $\alpha$ is greater than $35\%$. 
We hope that our in-depth incentive analysis for Bitcoin-NG can shed some light on the mechanism design and incentive analysis of next-generation blockchain protocols.

\bibliographystyle{plain}
\bibliography{bib}

\appendix
\section{Concentration Bounds}
\begin{lemma}[Chernoff bound for a sum of dependent random variables~\cite{niu2019analysis}]\label{lem:key_step}
Let $T$ be a positive integer. Let $X^{(j)} = \sum_{i = 0}^{n-1} X_{j + iT}$ be the sum of $n$ independent indicator random variables and $\mu_j = E\left( X^{(j)} \right)$ for $j \in \{1, \ldots, T\}$. Let $X = X^{(1)} + \cdots + X^{(T)}$. 
Let $\mu = \min_j \{ \mu_j \}$. Then, for $0 < \delta < 1$, $\Pr\left( X \le (1 - \delta) \mu T \right) \le e^{-\delta^2 \mu / 2}$.
\end{lemma}
\end{document}